\documentclass[journal]{IEEEtran}

\usepackage{xcolor,soul,framed} %,caption

\colorlet{shadecolor}{yellow}
\usepackage[pdftex]{graphicx}
\graphicspath{{../pdf/}{../jpeg/}}
\DeclareGraphicsExtensions{.pdf,.jpeg,.png}

\usepackage[cmex10]{amsmath}
\usepackage{array}
\usepackage{mdwmath}
\usepackage{mdwtab}
\usepackage{eqparbox}
\usepackage{url}
\usepackage{caption}
\usepackage{subfigure}
\newtheorem{lemma}{\textbf{Lemma}}  
\newtheorem{proof}{\textbf{Proof}} 
\usepackage[linesnumbered,ruled,vlined]{algorithm2e}

\begin{document}
\bstctlcite{IEEEexample:BSTcontrol}
    % \title{Resources Allocation based on Dynamic Pricing  in Mobile Edge Computing Task Offloading}
    \title{Dynamic Pricing based Near-Optimal Resource Allocation for Elastic Edge Offloading}
  \author{Yun Xia, Hai Xue, \IEEEmembership{Member IEEE}, Di Zhang, \IEEEmembership{Senior Member, IEEE}, Shahid Mumtaz, \IEEEmembership{Senior Member, IEEE},  Xiaolong Xu, \IEEEmembership{Senior Member, IEEE}, Joel J. P. C. Rodrigues, \IEEEmembership{Fellow, IEEE}

\thanks{This work was supported in part by the National Natural Science Foundation of China (NSFC) under Grant 62372242 and U22A0001, the Henan Natural Science Foundation for Excellent Young Scholar under Grant 242300421169, the Brazilian National Council for Scientific and Technological Development - CNPq via Grant No. 306607/2023-9. \emph{(Corresponding author: Hai Xue.)}

Yun Xia and Hai Xue are with the School of Optical-Electrical and Computer Engineering, University of Shanghai for Science and Technology, Shanghai 200093, China (e-mail: xiadayun99@gmail.com, hxue@usst.edu.cn).

Di Zhang is with the School of Electrical and Information Engineering, Zhengzhou University, the Henan International Joint Laboratory of Intelligent Health Information System, the National Telemedicine Center, and the National Engineering Laboratory for Internet Medical Systems and Applications, Zhengzhou 450001, China (e-mail: dr.di.zhang@ieee.org).

Shahid Mumtaz is with the Department of Engineering, Nottingham Trent University, Nottingham NG1 4FQ, UK, and the Department of Electronic Engineering, Kyung Hee University, Yongin-si, Gyeonggi-do 17104, South Korea (e-mail:
dr.shahid.mumtaz@ieee.org).

Xiaolong Xu is with the School of Software, Nanjing University of Information Science and Technology, Nanjing 210044, China (e-mail: xlxu@ieee.org).

Joel J. P. C. Rodrigues is with the Amazonas State University, Manaus - AM, Brazil (e-mail: joeljr@ieee.org).
}
  
      % Tibault~Reveyrand,~\IEEEmembership{Member,~IEEE,}\\
      % Ignacio~Ramos,~\IEEEmembership{Student Member,~IEEE,}
      % Erez Falkenstein,~\IEEEmembership{Student Member,~IEEE,}
      % and~Zoya~Popovi\'c,~\IEEEmembership{Fellow,~IEEE}% <-this % stops a space

%   \thanks{Manuscript received July 10, 2012. \hl{This paper is an expanded paper from the IEEE MTT-S Int. Microwave Symposium held on June 17-22, 2012 in Montreal, Canada.} This work was funded in part by the Office of Naval Research under the Defense Advanced Research Projects Agency (DARPA) Microscale Power Conversion (MPC) Program under Grant N00014-11-1-0931, and in part by the Advanced Research Projects Agency-Energy (ARPA-E), U.S. Department of Energy, under Award Number DE-AR0000216.}
%   \thanks{M. Roberg is with TriQuint Semiconductor, 500 West Renner Road Richardson, TX 75080 USA (e-mail: michael.roberg@tqs.com).}% <-this % stops a space
%   \thanks{T. Reveyrand is with the XLIM Laboratory, UMR 7252, University of Limoges, 87060 Limoges, France (e-mail: tibault.reveyrand@xlim.fr).}%
%   \thanks{I. Ramos and Z. Popovic are with the Department of Electrical, Computer and Energy Engineering, University of Colorado, Boulder, CO, 80309-0425 USA (e-mail: ignacio.ramos@colorado.edu; zoya.popovic@colorado.edu).}% <-this % stops a space
%   \thanks{E. Falkenstein is with Qualcomm Inc., 6150 Lookout Road
% Boulder, CO 80301 USA (e-mail: erez.falkenstein@gmail.com).}}  
}

% The paper headers
% \markboth{IEEE TRANSACTIONS ON MICROWAVE THEORY AND TECHNIQUES, VOL.~60, NO.~12, DECEMBER~2012
% }{Roberg \MakeLowercase{\textit{et al.}}: High-Efficiency Diode and Transistor Rectifiers}

% ====================================================================
\maketitle

% === ABSTRACT ====================================================================
% =================================================================================
\begin{abstract}
%\boldmath
% In mobile edge computing (MEC), task offloading can effectively reduce end user (EU) task execution latency and lower energy consumption. However, edge server (ES) resources are limited. When EUs offload tasks to the ES, it is essential to allocate server computing and bandwidth resources reasonably to maintain the sustainable and healthy development of MEC systems. This paper proposes a pricing mechanism for ES resources and establishes an EU utility function to regulate EU's demands on these resources. We have proven that the EU utility function has a local maximum. However, we found that as ES resource allocation increases, the EU utility function grows very slowly. To utilize ES resources more efficiently, this paper further explores sub-optimal resource allocation within a certain range near the local maximum. We propose a Dynamic Inertia and Speed-Constrained particle swarm optimization (DISC-PSO) algorithm to find these sub-optimal resource allocations. Extensive simulation results demonstrate the effectiveness of DISC-PSO in terms of the EU utility function, iteration count, and algorithm robustness.
In mobile edge computing (MEC), task offloading can significantly reduce task execution latency and energy consumption of end user (EU). However, edge server (ES) resources are limited, necessitating efficient allocation to ensure the sustainable and healthy development for MEC systems. In this paper, we propose a dynamic pricing mechanism based near-optimal resource allocation for elastic edge offloading. First, we construct a resource pricing model and accordingly develop the utility functions for both EU and ES, the optimal pricing model parameters are derived by optimizing the utility functions. In the meantime, our theoretical analysis reveals that the EU's utility function reaches a local maximum within the search range, but exhibits barely growth with increased resource allocation beyond this point. To this end, we further propose the Dynamic Inertia and Speed-Constrained particle swarm optimization (DISC-PSO) algorithm, which efficiently identifies the near-optimal resource allocation. Comprehensive simulation results validate the effectiveness of DISC-PSO, demonstrating that it significantly outperforms existing schemes by reducing the average number of iterations to reach a near-optimal solution by 92.11\%, increasing the final user utility function value by 0.24\%, and decreasing the variance of results by 95.45\%.

 % Extensive simulation results validate the effectiveness of DISC-PSO in optimizing resource allocation, as demonstrated by improvements in EU utility, iteration count, and the algorithm's robustness.

\end{abstract}

% === KEYWORDS ====================================================================
% =================================================================================

% For peer review papers, you can put extra information on the cover
% page as needed:
% \ifCLASSOPTIONpeerreview
% \begin{center} \bfseries EDICS Category: 3-BBND \end{center}
% \fi
%
% For peerreview papers, this IEEEtran command inserts a page break and
% creates the second title. It will be ignored for other modes.
\IEEEpeerreviewmaketitle

% ====================================================================
% ====================================================================
% ====================================================================

% === I. INTRODUCTION =============================================================
% =================================================================================
\section{Introduction}

\IEEEPARstart{I}{n} remote areas, deploying diverse sensor devices (SDs)  for environmental monitoring presents significant challenges due to the limited computational power and energy constraints of these SDs \cite{I1,I2}. When all data processing is conducted on-device, it rapidly depletes the SDs' batteries, thereby reducing their operational lifespan. To tackle this challenging issue, offloading the computational tasks to cloud servers has been widely investigated. However, the physical distance between the SDs and cloud servers causes significant delay \cite{I3,I4}, which is detrimental to delay-sensitive applications.

To mitigate the serious latency, mobile edge computing (MEC) has emerged as a promising computing paradigm. MEC enables the offloading of computational tasks to edge servers (ESs) situated closer to the SDs, which significantly reduces latency and energy consumption of the end devices \cite{I5,I6,I7}. In this way, it not only enhances the performance and efficiency of the sensor network but also extends the battery life of the SDs.

Nevertheless, ESs need incentives to participate in this process. Without appropriate compensation, ESs may be reluctant to offer their resources for task processing \cite{I71,I72}. Therefore, it is imperative to develop an efficient incentive mechanism to stimulate ESs to assist in task execution. Moreover, given the limited resources of ESs, adequately allocating these resources in the face of a large number of task offloading requests poses a significant challenge. Resource pricing as an incentive mechanism has garnered extensive attention in recent researches, since it has the potential to efficiently allocate resources and encourage participation from resource providers \cite{I8,I9,I10}. 

% In this context, a pricing-based incentive mechanism is proposed, where ESs set a price for their resources, and users (i.e., SDs) need to purchase these resources to offload their computational tasks. However, the resources of edge servers are also limited. When faced with a large number of task offloading, how to effectively allocate resources is also a challenge.

Therefore, in this paper, we explore the development of a dynamic pricing based resource allocation scheme, aiming to strike a balance between maximizing the utility of SDs while ensuring reasonable and fair resource allocation for ESs. To this end, we initially develop a dynamic pricing expression that relates ES resources to end users' (EU) data offloading amount. Based on the formulated pricing expression, a utility function for EU is constructed and its local maximum value is obtained by solving the utility function. Moreover, we theoretically analyzed and observed that the utility function exhibits barely growth near the local maximum value, indicating that limited benefits in allocating additional resources to EUs. Following that, a Dynamic Inertia and Speed-Constrained particle swarm optimization (DISC-PSO) algorithm is proposed to devise a near-optimal resource allocation scheme. Within a given precision range, DISC-PSO obtains near-optimal EU utility and derives corresponding near-optimal resource allocations. Furthermore, we prove that when the EU's utility function reaches the near-optimal level, so does the ES's utility, thus validating the effectiveness and rationality of the proposed DISC-PSO algorithm.

The key contributions of this paper are concluded as follows.
\renewcommand\labelenumi{\theenumi)} 
\begin{enumerate}
    \item We develop a dynamic pricing scheme that takes into account computing resources, bandwidth resources, and offloaded data amount. The proposed scheme exhibits an intuitive representation of the relationship between EU requirements and expenditures, while ensuring fair resource allocation among diverse EUs.
    \item We formulate a utility function to quantify the benefits that EUs can derive from task offloading and their corresponding expenditures. Furthermore, we demonstrate that this utility function exhibits a local maximum, indicating that EUs can achieve optimal utility by appropriately balancing offloaded data amount with the acquisition of computing and bandwidth resources.
    \item We theoretically analyzed and proved that the utility function grows barely while arriving near the local maximum value, even through allocating more resources. To this end, we propose the DISC-PSO algorithm to elucidate a near-optimal resource allocation scheme. Comprehensive simulation results demonstrate that the proposed DISC-PSO algorithm exhibits significant improvements compared to existing algorithms. Specifically, it reduced the number of iterations required to reach the near-optimal solution by an average of 92.11\%, increased the final EU utility function value by 0.24\%, and decreased the variance of the final results by 95.45\%.
    % \item We propose the DISC-PSO algorithm \textcolor{green}{to tackle the near-optimal} resource allocation among EUs. In comparison with PSO, DE, and GA, the DISC-PSO algorithm demonstrated a significant improvement across multiple experiments. On average, the number of iterations required to reach the near-optimal solution (i.e., EU utility function) was reduced by 92.11\%. Additionally, the final EU utility function value increased by an average of 0.24\%, and the variance of the final results was reduced by 95.45\%.}
    % By deploying this algorithm, ES resources can be allocated and utilized more rationally.
    % \item \textcolor{green}{Comprehensive simulation results and theoretical analysis are presented that the proposed dynamic pricing scheme not only rationalizes pricing but also enables the utilities of both the EU and ES to reach local maxima.}
    \item We conducted extensive simulations and theoretical analysis to validate the effectiveness of our proposed dynamic pricing scheme. The results reveal that it not only rationalizes resource pricing but also ensures the utilities of both EU and ES to reach their local maxima.

\end{enumerate}
% By integrating economic principles with technical strategies, this approach seeks to enhance the overall performance and sustainability of remote environmental monitoring systems.

The rest of this paper is organized as follows. In Section II, we present the related works. Section III illustrates the proposed system model which includes the model of time and energy consumption, also exhibits the problem formulation and solutions. In Section IV, simulation results are presented and discussed. The conclusion is given in Section V.

\section{Related work}
Over the past few years, resource allocation based on pricing schemes has been extensively investigated, employing various specific methodologies, we summarized them into the three following categories. 
\subsection{Optimization algorithm-based resource allocation}
A couple of researches \cite{R1,R8,R10,R12} have investigated the optimal resource allocation through the utilization of optimization algorithms. 
Huang $et$ $al.$ \cite{R1} proposed a hierarchical optimization algorithm for resource allocation, which decomposes the problem into three independent subproblems, and solves them by optimization algorithms such as genetic algorithm, and differential evolution algorithm.
% : 1). Task grouping problem is optimized using a genetic algorithm; 2). Pricing of computing resources is optimized using differential evolution (DE); 3). Pricing of energy is optimized using DE. Specifically, the service provider (SP) maximizes total profit by optimizing prices of computing resources and energy, while the device owner optimizes selection mode, broadcast power, and computing resource allocation based on prices provided by the SP. 
Zhou $et$ $al.$ \cite{R8} aimed to jointly optimize the decision of computation offloading, task offloading ratio, as well as the allocation of communication and computing resources, in order to minimize the overall cost of task processing for all EUs. The authors addressed this optimization problem by enhancing both the artificial fish swarm algorithm and the particle swarm optimization algorithm, resulting in efficient resource allocation for communication and computing.
Teng $et$ $al.$ \cite{R10} proposed a hybrid time-scale joint computation offloading and wireless resource allocation scheme, optimizing the offloading strategy and CPU frequency to achieve efficient resource allocation. They formulated a random convex optimization problem to optimize the offloading strategy and CPU frequency, utilizing the L0-norm as a representation of the offloading strategy. Moreover, the sample average approximation method was utilized to approximate the solution of this hybrid time-scale optimization problem. By adjusting both the offloading strategy and resource allocation scheme, the algorithm can effectively identify an optimal solution.
Xiang $et$ $al.$ \cite{R12} proposed a modeling approach for multi-stage convex optimization problems to achieve efficient resource allocation. The authors transformed the problem of cost-effective service provision into a multi-stage convex optimization problem, thereby optimizing the resource allocation. Additionally, they proposed an online algorithm RDC based on the Lyapunov framework, which decomposes the entire problem into multiple subproblems to facilitate resource allocation.

\subsection{Stackelberg game-based resource allocation}
Several papers \cite{R3,R7,R14} have employed Stackelberg game theory to achieve optimal resource allocation.
Cheng  $et$ $al.$ \cite{R3} conducted resource allocation within the framework of a Stackelberg game, and proposed two distinct schemes for pricing and resource allocation: one based on user granularity and the other on task granularity. By employing the backward induction method, edge servers determine the optimal strategy for pricing and resource allocation based on EUs' responses.
Ahmed $et$ $al.$ \cite{R7} introduced a dynamic resource allocation method based on Stackelberg game theory, enabling efficient utilization of multiple MEC servers. Specifically, the authors formulated a comprehensive Stackelberg game model to analyze pricing and resource procurement challenges faced by MEC servers. In addition, they devised an algorithm for multi-to-multi resource sharing among MEC servers. They also developed a separate algorithm for one-to-many resource allocation between MEC servers and edge service providers.
Li $et$ $al.$ \cite{R14} employed a multi-stage Stackelberg game for achieving resource allocation, wherein vehicles are incentivized to share their idle resources through reasonable contracts. Additionally, the interaction among roadside units, MEC servers, and mobile device EUs is modeled using the Stackelberg game to attain effective resource allocation and computation task offloading decisions. The optimal strategy for each stage is obtained via the backward induction method, thereby optimizing resource allocation.

\subsection{Auction-based resource allocation}
Auction mechanisms have also been utilized in existing researches \cite{R2,R4,R5,R9,R11,R13} for the purpose of resource allocation.
Habiba $et$ $al.$ \cite{R2} proposed a repeated auction model for resource allocation, wherein MEC servers and EU devices participated in bidding, and the final resource allocation was determined by the bidding results. To be specific, they designed a modified generalized second price algorithm to set pricing and allocate resources, taking into account the dynamic arrival of offloading requests and the computing workloads of ESs.
Sun $et$ $al.$ \cite{R4} devised a bid-based two-sided auction framework wherein mobile devices compete for computing services while edge servers offer such service through bidding. Within this framework, the authors also proposed two specific two-sided auction schemes: cost balancing and dynamic pricing, thereby achieving efficient resource allocation.
Bahreini $et$ $al.$ \cite{R5} proposed two auction-based resource allocation and pricing mechanisms, namely FREAP and APXERAP. FREAP employs a greedy approach, while APXERAP utilizes linear programming as an approximation technique. Both mechanisms consider the heterogeneity of EU requests and resource types, resulting in resource allocation and pricing strategies that closely approximate optimal social welfare.
Wang $et$ $al.$ \cite{R9} proposed an auction-based resource allocation mechanism that considers the flexibility and locality of task offloading, enabling joint allocation of computing and communication resources as well as partial offloading. The mechanism achieves efficient resource allocation and enhances social welfare by integrating a minimum delay task graph partitioning algorithm, a feasible non-dominated resource allocation determination algorithm, and an original-dual approximate winning bid selection algorithm.
Li $et$ $al.$ \cite{R11} utilized a two-sided auction mechanism to achieve resource allocation. Specifically, the authors proposed two two-sided auction mechanisms, namely TMF and EMF, to optimize both rationality in resource allocation and social welfare considering the successful transactions.
Li $et$ $al.$ \cite{R13} proposed a double auction game model to address the resource management problem in MEC and Internet of Things (IoT) devices. In this mechanism, MEC acts as seller while IoT devices act as buyers, submitting their bids and offers. The auction intermediary determines the transaction price based on participants' submissions, establishing a fair and effective market for resource transactions. The experience-weighted attraction algorithm is proposed to enhance the learning process of transaction strategies and find the Nash equilibrium point.

\begin{figure}[!tp]
\captionsetup{singlelinecheck = false, justification=justified}
\centering
\includegraphics[width=3in]{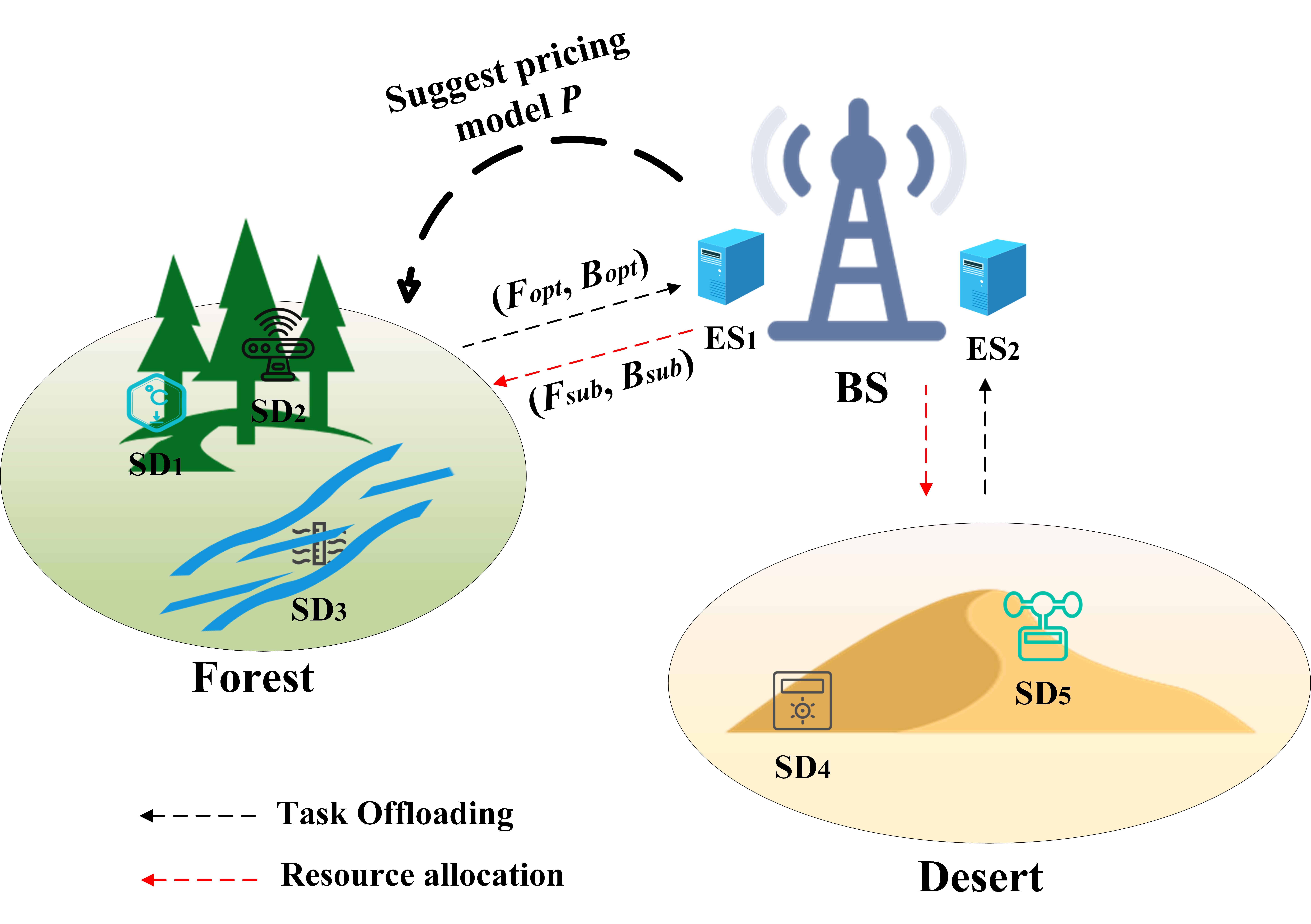}
\caption{System model. }
\label{fig11}
\end{figure}

% In addition, Liu $et$ $al.$ \cite{R6} employed a resource allocation approach grounded in microeconomic theory, wherein they devised optimal budget allocation algorithms and equilibrium price discovery algorithms to achieve a harmonious balance between the supply and demand of resources within a mobile edge computing system.

However, current pricing-based resource allocation schemes encounter certain issues. Existing researches \cite{R3,R4,R11,R13} simplify the representation of resources and assign abstract prices to them, which may facilitate problem-solving but fail to accurately reflect the actual and precise resource requirements of EUs. Moreover, both optimization algorithm-based resource allocation and Stackelberg game-based resource allocation suffer from the drawback that the price is determined during the optimization process or game \cite{R1,R8,R10,R12,R3,R7,R14}, thereby rendering it transparent to EUs before their purchase. Consequently, this fails to align with the genuine resource needs of EUs. In the auction process described in \cite{R2,R4,R5,R9,R11,R13}, although EUs are aware of their bids, multiple rounds of auctions are required for finalizing the resource allocation outcome, inevitably resulting in serious time consumption.

\section{system model and problem formulation}
\subsection{System model}
As depicted in Fig. \ref{fig11}, we consider the resource allocation problem between single EU and ES in remote area (e.g., forest, desert). It is known that when the EU offloads data to the ES, it can reduce energy consumption. The more ES resources they purchase (e.g., computational frequency $F$ and bandwidth $B$), the faster the tasks are processed and returned. Therefore, the EU needs to decide how much  $F$ and  $B$ to purchase. First, the ES provides a pricing model of $P$. Then, the EU constructs a utility function $U_{user}(P,F,B)$ and maximizes it to determine the optimal amount of resources ($F_{opt}$, $B_{opt}$) to purchase. However, when the offloaded data amount is fixed, the increase of EU's utility becomes negligible as the purchased resources increase\footnote{Considering that the amount of offloaded data is constant, initial resource increasement results in a substantial reduction in latency and energy consumption, leading to a noticeable improvement in utility. However, as more resources are supplied, the impact of these reductions becomes less pronounced, and the marginal gains in utility diminish.}. To rationalize resource allocation, the ES, based on the EU-provided ($F_{opt}$, $B_{opt}$) will determine a near-optimal resource allocation ($F_{sub}$, $B_{sub}$) within an acceptable range of the maximum utility function $U_{user}^{max}$. For the sake of convenience, Table \ref{table_2} summarizes the important notations in this paper.

\begin{table}[!t]
\renewcommand{\arraystretch}{1.3}
\caption{SUMMARY OF  PRIMARY NOTATIONS}
\label{table_2}
\centering

% \begin{tabular}{ c | >{\centering\arraybackslash}m{7cm}}
\begin{tabular}{ c | >{\arraybackslash}m{7cm}}
\hline

\textbf{Symbol} & \multicolumn{1}{c}{\textbf{Description}} \\
\hline
    $q$ &  The maximum amount of data offloading of the EU\\
\hline
    $c$ &  The number of cycles required to process a bit\\
\hline
    $F_{local}$ &  The local CPU frequency of EU\\
\hline
    $T_{local}$ & The local time of processing  data\\
\hline
$B$ & The bandwidth resources purchased by EU\\
\hline
$R_u$ & The speed of uplink transmission \\
\hline
    $R_d$ &  The  speed of downlink transmission\\
\hline
    $T_u$ &  The time of uplink transmission\\
\hline
    $T_d$ &  The time of downlink transmission\\
\hline
    $T_p$ &  The time of processing data on ES\\
\hline
    $\alpha$ & The data conversion ratio\\
\hline
% $t}$ & The correction factor\\
% \hline
    $k$ & The effective switched capacitance\\
\hline
    $F_{server}$ & The ES CPU frequency purchased by EU\\
\hline
    $E_{up}$ & The amount of energy consumed by uploading task\\
\hline
    $E_d$ & The amount of energy consumed by downloading  task\\
\hline
    $w_1$ & The discount factor of time-saving\\
\hline
    $w_2$ & The discount factor of energy-saving\\
\hline
\end{tabular}
\end{table}

\subsection {Time consumption model}
We define the amount of EU's offloading data as $q$, the number of CPU cycles required to process each bit of data as $c$, and the EU's local CPU frequency as $F_{local}$. Therefore, the EU's local data execution time is calculated as follows.
\begin{equation}
    T_{local}=\frac{qc}{F_{local}}.
\end{equation}

Since there are different signal-to-noise ratios $(S/N)_{uplink}$ and $(S/N)_{downlink}$ for the uplink and downlink transmission, different uplink and downlink transmission speed $R_u$ and $R_d$ is obtained respectively according to Shannon's formula.
%假设上下行链路的速率都为R
\begin{equation}
    R_u=Blog_2(1+(\frac{S}{N})_{uplink}),
    \label{RUP}
\end{equation}
\begin{equation}
    R_d=Blog_2(1+(\frac{S}{N})_{downlink}),
    \label{RD}
\end{equation}
where $B$ is the bandwidth purchased by the EU. Thus the uplink transmission time $T_u$ is expressed as follows.
\begin{equation}
    T_u=\frac{q}{R_u}.
\end{equation}

Assuming that the CPU frequency of ES computation purchased by the EU is $F_{server}$. Then, the time $T_p$ of data processing on the ES is expressed as follows.
\begin{equation}
    T_p=\frac{qc}{F_{server}}.
\end{equation}

Since the download data is relatively less compared to the amount of uploaded data, we assume the conversion ratio as $\alpha$ \cite{MKMH}. Therefore, the downlink transmission time $T_d$ is illustrated as follows.
%数据转换比为α
\begin{equation}
    T_d=\frac{\alpha q}{R_d}.
\end{equation}

To sum up, the total time of the entire process for EU offloading data to the ES includes the uplink transmission time $T_p$, the downlink transmission time $T_d$, and the execution time $T_p$ of the ES. Accordingly, the total offloading time  $T_{offload}$ is calculated as follows,
\begin{align}
    T_{offload}&=T_u+T_p+T_d \nonumber \\
    &=\frac{q}{R_u}+\frac{qc}{F_{server}}+\frac{\alpha q}{R_d} \nonumber \\
    &=q\cdot (\frac{1}{R_u}+\frac{c}{F_{server}}+\frac{\alpha }{R_d}).
    \label{Toffload}
\end{align}

%节省的T_s的时间为T_l-T_p-T_u-T_d
% \begin{align}  
%     T_s &= T_l - T_u - T_p - T_d \nonumber \\  
%         &= \frac{qr}{F_l}- \frac{q}{R}-\frac{qr}{F_c^i}-\frac{\alpha q}{R} \nonumber \\
%         &= q \cdot (\frac{r}{F_l}-\frac{1}{R}-\frac{r}{F_c^i}-\frac{\alpha}{R})     \label{your-equation-label}  
% \end{align}  

Therefore, the time $T_{save}$ saved by offloading versus non-offloading is illustrated as follows.
\begin{align}
    T_{save}&=T_{local}-T_{offload} \nonumber \\ 
    &=\frac{qc}{F_{local}}-q\cdot (\frac{1}{R_u}+\frac{c}{F_{server}}+\frac{\alpha }{R_d}) \nonumber \\
    &=q\cdot (\frac{c}{F_{local}}-\frac{1}{R_u}-\frac{c}{F_{server}}-\frac{\alpha }{R_d}).
    \label{Tsave}
\end{align}

\subsection {Energy consumption model}
If the EU processes $q$ bit data locally, the required energy consumption $E_{local}$ is calculated by the following expression \cite{RCPYDD}.
%k_l=10^(-26)
\begin{equation}
    E_{local}=k \cdot (qc) \cdot F_{local}^2,
\end{equation}
where $k$ is an energy consumption coefficient that depends on the chip architecture. The energy consumed by EU in the offloading process includes the energy consumption $E_{up}$ of uploading data, and the energy consumption $E_d$ of downloading data.
\begin{equation}
    E_{up}=P_u \cdot T_u,
\end{equation}
\begin{equation}
    E_{d}=P_d \cdot T_d,
\end{equation}
where $P_u$ and $P_d$ are the power of the EU for uploading and downloading data, respectively. Therefore, the energy $E_{save}$ saved by the EU in the offloading process is shown as follows.
%用户节省的能耗
\begin{align}
    E_{save}&=E_{local}-E_{up}-E_{d} \nonumber \\
    &=k \cdot (qc) \cdot F_{local}^2-P_u \cdot T_u-P_d \cdot T_d \nonumber \\
    &=q\cdot(kcF_{local}^2-\frac{P_u}{R_u}-\frac{P_d\cdot \alpha}{R_d}).  
    \label{Esave}
\end{align}

% %k_c=10^-4j/bit
% \begin{equation}
%     E_c=K_c \cdot q
% \end{equation}
% %p=0.1
% %用户上传的时间，消耗的能耗
% \begin{equation}
%     W_u=P_u \cdot T_u
% \end{equation}
% %p=1W
% %用户下载的时间，消耗的能耗
% \begin{equation}
%     W_d=P_d \cdot T_d
% \end{equation}

% %用户节省的能耗
% \begin{align}
%     W_s&=E_l-W_u-W_d \nonumber \\
%     &=k_l \cdot (qr) \cdot F_l^2-P_u \cdot T_u-P_d \cdot T_d \nonumber \\
%     &=k_l \cdot (qr) \cdot F_l^2-P_u\cdot \frac{q}{R}-P_d\cdot \frac{\alpha\cdot q}{r} \nonumber \\
%     &=q\cdot(k_lrF_l^2-\frac{P_u}{R}-\frac{P_d\cdot \alpha}{r}   
% \end{align}

\subsection {Validation function of EU}
As the original intention of edge offloading, EUs offload tasks to save energy and reduce latency. However, during this process, EUs must also cover the cost of acquiring resources from the ES to ensure proper functioning and support. Thus we can define the EU's utility function $U_{user}$ as follows.
\begin{equation}
    U_{user}=w_1 \cdot E_{save} + w_2 \cdot T_{save} -P,
    \label{userutility}
\end{equation}
where $w_1$ and $w_2$ denote the discount factors of time-saving and energy-saving \cite{WWW}, respectively, which denote the EU in favor of time or energy consumption. In addition, the payment function $P$ is positively correlated to the purchase of bandwidth resource $B$ and the computing frequency $F_{server}$ \cite{HHJKS}. For convenience, we let the payoff function $P$ satisfy the following expression.
\begin{equation}
    P=a\cdot F_{server}+b\cdot B,
    \label{Pay}
\end{equation}
where $a > 0$ and $b > 0$. Then, substituting Eq.(\ref{RUP}), Eq.(\ref{RD}), Eq.(\ref{Tsave}), Eq.(\ref{Esave}), and Eq.(\ref{Pay}) into Eq.(\ref{userutility}), we obtain the following expression.
\begin{align}
    U_{user}&=w_1(q\cdot(kcF_{local}^2-\frac{P_u}{Blog_2(1+(\frac{S}{N})_{uplink})} \nonumber\\
    &-\frac{P_d\cdot \alpha}{Blog_2(1+(\frac{S}{N})_{downlink})}))+w_2(q\cdot (\frac{c}{F_{local}}\nonumber\\
    &-\frac{1}{Blog_2(1+(\frac{S}{N})_{uplink})}-\frac{c}{F_{server}}\nonumber \\ &-\frac{\alpha }{Blog_2(1+(\frac{S}{N})_{downlink})})) \nonumber \\
    &-(a\cdot F_{server}+b\cdot B) \nonumber\\
    &=q \cdot (w_1kcF_{local}^2+\frac{w_2c}{F_{local}})-\frac{q}{F_{server}}\cdot w_2c \nonumber\\
    &- \frac{q}{B}(\frac{w_1P_u+w_2}{log_2(1+(\frac{S}{N})_{uplink})}+\frac{w_1P_d\alpha+w_2\alpha}{log_2(1+(\frac{S}{N})_{downlink})})\nonumber\\
    &-a F_{server}-bB.
    \label{U1}
\end{align}

To simplify the expression of $U_{user}$, we let 
\begin{equation}
    \mathcal{X} = w_1kcF_{local}^2+\frac{w_2c}{F_{local}}, 
\end{equation}
\begin{equation}
    \mathcal{Y} = \frac{w_1P_u+w_2}{log_2(1+(\frac{S}{N})_{uplink})}+\frac{w_1P_d\alpha+w_2\alpha}{log_2(1+(\frac{S}{N})_{downlink})}. 
    \label{yexpression}
\end{equation}

Substituting $\mathcal{X}$ and $\mathcal{Y}$ into Eq.(\ref{U1})  yields the following expression.   
\begin{equation}
    U_{user}=q\cdot \mathcal{X}-\frac{q}{F_{server}} \cdot w_2c - \frac{q}{B} \cdot \mathcal{Y}-a F_{server}-bB.
    \label{halfuser}
\end{equation}

We take the first partial and second partial derivations for $F_{server}$ and $B$, respectively. Then, we obtain the following expressions.
\begin{equation}
    \frac{\partial U_{user}}{\partial F_{server}} = \frac{w_2cq}{F_{server}^2}-a,
    \label{linjie1}
\end{equation}
\begin{equation}
    \frac{\partial U_{user}}{\partial B}=\frac{q\mathcal{Y}}{B^2}-b,
    \label{linjie2}
\end{equation}

\begin{equation}
    \frac{\partial^2U_{user}}{\partial^2F_{server}}=-\frac{2w_2cq}{F_{server}^3},
    \label{Hessian1}
\end{equation}
\begin{equation}
    \frac{\partial^2U_{user}}{\partial^2B}=-\frac{2q\mathcal{Y}}{B^3},
    \label{Hessian2}
\end{equation}
\begin{equation}
    \frac{\partial^2U_{user}}{\partial F_{server} \partial B}=0,
    \label{Hessian3}
\end{equation}
\begin{equation}
    \frac{\partial^2U_{user}}{\partial B \partial F_{server} }=0.
    \label{Hessian4}
\end{equation}

The Hessian matrix is classically utilized to determine the critical points of a function for optimization issues \cite{HSJZ1, HSJZ2, HSJZ3}. Therefore, we construct the Hessian matrix $\mathcal{H}$ using Eq.(\ref{Hessian1}), Eq.(\ref{Hessian2}), Eq.(\ref{Hessian3}), and Eq.(\ref{Hessian4}).
\begin{align}
  \mathcal{H} &=   
\begin{bmatrix}  
\frac{\partial^2U_{user}}{\partial^2F_{server}} & \frac{\partial^2U_{user}}{\partial F_{server} \partial B} \\  
\frac{\partial^2U_{user}}{\partial B \partial F_{server} } & \frac{\partial^2U_{user}}{\partial^2B}  
\end{bmatrix}  \nonumber \\
&=
\begin{bmatrix}  
-\frac{2w_2cq}{F_{server}^3} & 0 \\  
0 & -\frac{2q\mathcal{Y}}{B^3} 
\end{bmatrix}.
\end{align}

\begin{lemma}
    \label{lemma1}
    $\mathcal{H}$ is a negative definite matrix.
\end{lemma}
    
\begin{proof}
    please see the appendix \ref{negative} for details. 
\end{proof}

From Lemma \ref{lemma1}, we can know that the Hessian matrix $\mathcal{H}$ of $U_{user}$ is a negative definite matrix, which means the $U_{user}$ has a local maximum value at critical point\footnote{The local maximum within the search space is treated as a global maximum in our analysis, since there is no global maximum of $U_{user}$, resulting in the globally optimal allocation.}. We let Eq.(\ref{linjie1}) and Eq.(\ref{linjie2}) equal to 0, and then we can find the critical point ($F_{server}$, $B$) as follows. 
\begin{equation}
    \left\{ \begin{array}{l}  
    F_{server} = \sqrt{\frac{w_2cq}{a}}, \\  
    B = \sqrt{\frac{q\mathcal{Y}}{b}}.  
\end{array} \right.
\label{criticalpoint}
\end{equation}
\begin{lemma}
    \label{lemma2}
     $U_{user}$ hardly increases while approximating the local maximum.
\end{lemma}
\begin{proof}
    please see the appendix \ref{slow} for details.
\end{proof}

% From Eq.(30), since the Hessian matrix is greater than zero and the entries on the main diagonal are less than zero, this means that the \textbf{critical point of $U_{user}$ is a local maximum}. 
% We let Eq.(\ref{linjie1}) and Eq.(\ref{linjie2}) equal to 0, and then we can find the \textbf{critical point} ($F_{server}$, $B$) as shown in follows. 
% \begin{equation}
%     \left\{ \begin{array}{l}  
%     F_{server} = \sqrt{\frac{w_2cq}{a}} \\  
%     B = \sqrt{\frac{q\mathcal{Y}}{b}}  
% \end{array} \right.
% \label{criticalpoint}
% \end{equation}

Eq.(\ref{criticalpoint}) illustrates that when the EU offloads $q$ bit data, how many computational frequency and bandwidth resources of the edge server to be purchased, so as to maximize its own utility function. Therefore, by substituting Eq.(\ref{criticalpoint}) into Eq.(\ref{halfuser}), the final expression for $U_{user}$ is obtained.
\begin{equation}
    U_{user}=q\cdot \mathcal{X}-\frac{2q}{F_{server}} \cdot w_2c - \frac{2q}{B} \cdot \mathcal{Y}.
    \label{userutility0}
\end{equation}

% In addition,  we plot the user utility function $U_{user}$ for different $F_{server}$ and $B$, as shown in Fig. \ref{fig2}. We set the bandwidth resource purchased by the user as $B \in [0.1Mbps, 1Mbps]$, the computation frequency as $F \in [1GHz, 6GHz]$, and the amount of offloaded data as $500KB$. As can be seen in Fig. \ref{fig2}, the $U_{user}$ has a local maximum at the point $(1Mbps, 6GHz)$.
% \begin{figure}[!th]
% \captionsetup{singlelinecheck = false, justification=justified}
% \centering
% \includegraphics[width=3in]{fig2.png}
% \caption{Effect of different $F_{server}$ and $B$ on  $U_{user}$. }
% \label{fig2}
% \end{figure}
\subsection{Validation function of ES}
%处理任务时间越长需要保持连接，加剧了成本
In edge computing environment, offloading tasks to the ES necessitates the use of ES resources to process the computations. Since the ES requires compensation for providing these resources, the EU must pay fees to offload its task effectively. At the same time, for the ES, it takes time to process the EU's task, which means that the ES has a time cost\footnote{Here, even if ES does not process EU data, it still requires a significant amount of energy to operate. The energy expended in processing EU data is relatively small compared to this, so we ignore energy consumption costs. Additionally, some papers \cite{R7,energy1,energy2} also do not include energy consumption as part of the ES utility function.}. The less time the ES spends on processing EU tasks, the lower the cost. In addition, when the ES processes the EU's data, it also obtains the EU's data, and these data also bring revenue $W$ to the ES \cite{revenue}. So the utility function of the ES is obtained as follows.
\begin{equation}
    U_{server}=P-T_{offload}+W.
    \label{Userver}
\end{equation}

We substitute Eq.(\ref{criticalpoint}) into Eq.(\ref{Pay}) to obtain the expression for $P$ as follows.
\begin{equation}
    P=\frac{w_2cq}{F_{server}}+\frac{q\mathcal{Y}}{B}.
    \label{Payplus}
\end{equation}

Inspired by \cite{Inspired1} and \cite{Inspired2}, we define $W$ as the following expression.
\begin{equation}
    W=\mu log_2(1+q),
    \label{Wdata}
\end{equation}
where $\mu $ $\in$ (0, 1) denotes the reward index. Substituting Eq.(\ref{Toffload}), Eq.(\ref{Payplus}) and Eq.(\ref{Wdata}) into Eq.(\ref{Userver}) yields the following expression.
\begin{equation}
    \begin{aligned}
    U_{server}&=\frac{w_2cq}{F_{server}}+\frac{q\mathcal{Y}}{B}-q\cdot (\frac{1}{R_u}+\frac{c}{F_{server}}+\frac{\alpha }{R_d})\\&+\mu log_2(1+q) \\
    &=\frac{cq(w_2-1)}{F_{server}}+\frac{q}{B}(\frac{w_1P_u+w_2-1}{log_2(1+(\frac{S}{N})_{uplink})}\\&+ \frac{w_1P_d\alpha+w_2 \alpha-\alpha}{log_2(1+(\frac{S}{N})_{downlink})})+\mu log_2(1+q).
    \end{aligned}
    \label{Userver1}
\end{equation}
% It also can be seen from Fig. \ref{fig_6a} that the user utility function does not change much near the maximum point. In other words, users buy more resources and the utility function does not increase much. For the server, similarly the utility function does not increase much as shown in Fig. \ref{fig_6c}.

\begin{table}[!t]
\renewcommand{\arraystretch}{1.3}
\caption{\centering{SUMMARY OF  PRIMARY NOTATIONS FOR ALGORITHM 1}}
\label{table_1}
\centering

\begin{tabular}{ c | >{\arraybackslash}m{5.5cm}}
\hline

\textbf{Symbol} & \multicolumn{1}{c}{\textbf{Description}} \\
\hline
    $P_n$ &  The number of particles \\
\hline
    $w_{max}$ &  The maximum inertia weight\\
\hline
    $w_{min}$ &  The minimum inertia weight\\
\hline
    $c_1$ & The individual learning factor\\
\hline
    $c_2$ & The social learning factor\\
\hline
    $\Delta$$F$ & The minimum change value of $F_{server}$\\
\hline
    $\Delta$$B$ & The minimum change value of $B$\\
\hline
    $[F_{server}^{min}$, $F_{server}^{max}]$ &  The search range of $F_{server}$\\
\hline
    $[B_{min}, B_{max}]$ &  The search range of $B$\\
\hline
    $N_t$ &  The number of experiments\\
\hline
    $N_{max}$ & The maximum number of iterations\\
\hline
    $N_f$ & The  number of actual iterations\\
\hline
    $P_{pos}$ & The position of  particles\\
\hline
    $V$ &   The velocity of  particles\\
\hline
    $P_{pb}$ & The best position for each particle\\
\hline
    $P_{bs}$ & The best $U_{user}$ for each particle\\
\hline
    $S_{gb}$ & The  largest $U_{user}$ of all particles \\
\hline
    $I_{bp}$ & The  index corresponding to $S_{gb}$\\
\hline
    $P_{gb}$ & The  best position of all particles\\
\hline
\end{tabular}
\end{table}

\subsection{Server optimal policy}
As aforementioned, Lemma \ref{lemma2} proves that $U_{user}$ hardly increases while approximating the local maximum. Given the efficient and rational utilization of server resources, we pick the near-optimal resource allocation. When the utility functions of both ES and EU remain fairly constant within a certain resource allocation range, we allocate fewer resources, which is beneficial for saving resources and allowing the ES to serve more EUs over a while. We define the allocation strategy ($F_{sub}$, $B_{sub}$) as a near-optimal allocation if the condition $\frac{U_{max}-U_{sub}}{U_sub}< \epsilon$ is satisfied, where $U_{max}$ is the optimal utility within the search space, $U_{sub}$ is the utility achieved under the allocation scheme ($F_{sub}$, $B_{sub}$), and $\epsilon$ is a predefined threshold value.  For this reason, we propose a DISC-PSO algorithm as illustrated in algorithm \ref{alg:alg1} to seek this near-optimal resources allocation point ($F_{server}$, $B$) through continuous iterations. Primary notations involved in algorithm \ref{alg:alg1} are summarized in Table \ref{table_1}. First, we obtain the EU's maximum utility $U_{user}^{max}$ through Eq.(\ref{criticalpoint}) (i.e., the critical point). Initialization is performed at the same time, including the $P_{pos}$, $V$, $P_{pb}$, and $A_s[P_n][1]$ for all particles, while recording $S_{gb}$ and $P_{gb}$ in all particles (see line \ref{step1} to \ref{step8} of algorithm \ref{alg:alg1}).  It then goes into an iteration, updating the $V$ and $P_{pos}$, while ensuring that $\Delta F$ and $\Delta B$ are satisfied (see line \ref{step12} to \ref{step21} of algorithm \ref{alg:alg1}). Then we update the $P_{pb}$ and $A_s[P_n][1]$. If there are values in the matrix $A_s[P_n][1]$ greater than  $S_{gb}$, then we update the $P_{gb}$ and $S_{gb}$ (see line \ref{step23} to \ref{step33} of algorithm \ref{alg:alg1}). Subsequently, it is judged whether the termination condition as shown in Eq.(\ref{condition}) is satisfied. If so, the iteration process is terminated (see line \ref{step35} to \ref{step36} of algorithm \ref{alg:alg1}). Finally, the experimental results included $S_{gb}$, $N_f$, and $P_{gb}$ are recorded.

\begin{equation}
    \frac{U_{user}^{max}-S_{gb}}{S_{gb}} < \epsilon.
    \label{condition}
\end{equation}
% where $\epsilon$ is a small threshold value.}

In comparison with the traditional PSO algorithm, our proposed DISC-PSO algorithm is enhanced from the following aspects to accelerate the convergence speed: 1). We change the inertia weight $w$ from fixed value to dynamic change (see line \ref{new1} of algorithm \ref{alg:alg1}); 2). We add the condition limiting the minimum change in speed (see line \ref{new2} to line \ref{new3} of algorithm \ref{alg:alg1}); 3). In addition to reaching the maximum number of iterations to stop iteration, we also add the termination condition of Eq.(\ref{condition}) (see line \ref{step35} to line \ref{step36} of algorithm \ref{alg:alg1}). The time complexity of the DISC-PSO algorithm is analyzed as follows:
\begin{itemize}  
\item The outer loop runs for $N_t$ iterations, which represents the number of experiments conducted.
\item The inner loop (i.e., lines 10-36) is executed up to $N_{max}$ times, representing the maximum number of iterations for particle updates.
\item Within the inner loop, the for-loop (i.e., lines 11-21) iterates over $P_n$ particles, and each particle's position and velocity are updated in $O(1)$ time. The complexity for this section is $O(P_n)$.
\item The update of personal and global bests (i.e., lines 23-33) also requires $O(P_n)$.
\item The stopping condition check (i.e., line 35) takes $O(1)$ time.
\end{itemize}

Given the nested loops, the overall time complexity of the algorithm is $O(N_t \times N_{max} \times P_n)$, which represents the worst-case scenario where the algorithm runs through all possible iterations.

\begin{table}[h]
    \centering
    \caption{DEFAULT PARAMETER SETTINGS}
    \begin{tabular}{|c|>{\centering\arraybackslash}m{3.2cm}|c|>{\centering\arraybackslash}m{0.8cm}|}
         \hline  
        \textbf{Parameter}  & \textbf{Value}  & \textbf{Parameter} & \textbf{Value}  \\ 
        \hline  
        $c$ & 2640 $cycles/bit$ & $P_d$ & 1$W$ \\ 
        \hline 
        $F_{server}$ & $\{1, 2, ..., 6\} GHz$  & $P_u$ & 0.1$W$ \\
        \hline
        $\mu$ & 0.8 &$w_1$, $w_2$ &0.5  \\
        \hline
        $F_{local}$ & $\{0.1, 0.2, ..., 1\} GHz$  & $(S/N)_{uplink}$ & 20$dB$\\
        \hline
        $q$ &[100, 500] $KB$ & $(S/N)_{downlink}$ & 30$dB$\\
        \hline
        $k$ & $10^{-27} Watt\cdot s^3/cycles^3$  &$\alpha$& 0.2\\
        \hline
        $B$ & \{0.1, 0.2, ..., 1\}$Mbps$  & $\epsilon$ & 0.001\\
        \hline
    \end{tabular}
    \label{tabel_3}
\end{table}

\begin{algorithm}[!t]
    \SetAlgoLined
    \caption{DISC-PSO}\label{alg:alg1}
    \KwIn{  $P_n$,  $w_{max}$,  $w_{min}$, $c_1$, $c_2$,  $\Delta$$F$,  $\Delta$$B$,  [$F_{server}^{min}$, $F_{server}^{max}$],  [$B^{min}$, $B^{max}$], $N_t$, $N_{max}$, $\epsilon$}
    \KwOut{$S_{gb}$,  $N_f$, $P_{gb}$ }

    \For{exp = 1: $N_t$}{
        Substitute Eq.(\ref{criticalpoint}) into Eq.(\ref{userutility})  to obtain the maximum value $U_{user}^{max}$ \label{step1}\;
        $P_{pos}$ = [$F_{server}^{min}$ + ($F_{server}^{max}$ - $F_{server}^{min}$) $\times$ rand($P_n$, 1), $B^{min}$ + ($B^{max}$ - $B^{min}$) $\times$ rand($P_n$, 1)]\;
        $V$ = zeros($P_n$, 2)\;
        $P_{pb}$ $\leftarrow$ $P_{pos}$\;
        Calculate $U_{user}$ for the current position of each particle and record it in the matrix $A_s[P_n][1]$\;
        $[S_{gb}, I_{bp}]$ $\leftarrow$ $\max$ ($A_s[P_n][1]$)\;
        $P_{gb}$ = $P_{pb}(I_{bp}, :)$ \label{step8}\;
        $N_f = 0$\;
            \While{$N_f < N_{max}$}{
                \For{$i = 1: P_n$}{
                $r_1$ = rand \label{step12}\;
                $r_2$ = rand\;
                $w=w_{max} -((w_{max}-w_{min}) \times N_f/ N_{max})$ \label{new1}\;
                $V(i, :) = w \times V(i, :) + c_1 \times r_1 \times
                (P_b(i, :) - P_{pos}(i, :)) + c_2 \times r_2 \times (P_{gb} - P_{pos}(i, :))$\;
            
                % 确保速度满足最小变化量
                $V(i, 1) = sign(V(i, 1)) \times \max(abs(V(i, 1)), \Delta F$) \label{new2}\;
                $V(i, 2) = sign(V(i, 2)) \times \max(abs(V(i, 2)), \Delta B$) \label{new3}\;

                $P_{pos}(i, :) = P_{pos}(i, :) + V(i, :)$\;
                $P_{pos}(i, 1) = \max(F_{server}^{min}$\; $\min(F_{server}^{max}, P_{pos}(i, 1)))$\;
                $P_{pos}(i, 2) = \max(B^{min}, \min(B^{max}, P_{pos}(i, 2)))$ \label{step21}\;
                }
                % 更新个人和全局最优
                $S = A_s[P_n][1]$ \label{step23}\;
                \For{$i = 1: P_n$}{
                \If{$S(i) > A_s[P_n][1](i)$ }{
                    $A_s[P_n][1](i)=S(i)$\;
                    $P_{pb}(i, :) = P_{pos}(i, :)$\;
                    }
                } 
                $[S_{ngb}, I_{bp}] = \max (A_s[P_n][1])$ \label{step30}\;
                \If{$S_{ngb} > S_{gb}$ }{
                    $S_{gb} = S_{ngb}$\;
                    $P_{gb} = P_{pb}(I_{}bp, :)$ \label{step33}\;
                }
                % 检查停止条件
                \If{$abs(S_{gb}-U_{user}^{max})/S_{gb} < \epsilon$ \label{step35}}{
                    break \label{step36}\;
                }
                $N_f = N_f$ + 1\;
            }
        %记录每次结果
        Record $S_{gb}$, $P_{gb}$, and $N_f$\;
        }
\end{algorithm}

\section{performance evaluation}
In this section, in order to verify the effectiveness and rationality of the proposed dynamic pricing scheme, we explore the pricing (i.e., $P$), the corresponding $U_{user}$ and $U_{server}$ with varying resource allocations (i.e., $F_{server}$ and $B$). Furthermore, we analyze the near-optimal value of the EU utility function with the proposed DISC-PSO algorithm, which demonstrates advantages in terms of efficiency, stability, and solution quality compared to existing algorithms. It should be noted that optimizing the EU utility function inherently results in the optimization of the ES utility function, as proven in Appendix \ref{delte}. This ensures that the proposed resource allocation scheme effectively optimizes both utilities, allowing us to focus primarily on the EU utility in the subsequent analysis without reiterating the equivalence.
\subsection{Simulation Setup}
For the simulation setup, the local CPU frequency $F_{local}$ of each EU is uniformly selected from the set \{0.1, 0.2, ..., 1\} $GHz$, $\alpha$ = 0.2, $P_u$ =0.1$W$, and $P_d$ = 1$W$ \cite{MKMH}.  In addition, the data size of EU is uniformly distributed with  $q$ $\in$ [100, 500] $KB$ \cite{data}. The remaining parameter settings are summarized in Table \ref{tabel_3} \cite{HHJKS,MY,LMZTQSQ}.

\subsection{The effect of $F_{server}$ on $P$, $U_{user}$, and $U_{server}$.}
As shown in Fig. \ref{fig1}, we explore the resource pricing $P$, the EU utility function $U_{user}$, and the ES utility function $U_{server}$ with different $F_{server}$ purchased by EUs. We set the bandwidth resource of the EU to buy $B = 0.1Mbps$, as well as the amount of offloaded data $500KB$. In addition, we denote $F$ vary from $1GHz$ to $6GHz$.

\begin{figure}[th]
\captionsetup{singlelinecheck = false, justification=justified}
\centering
\includegraphics[width=3in]{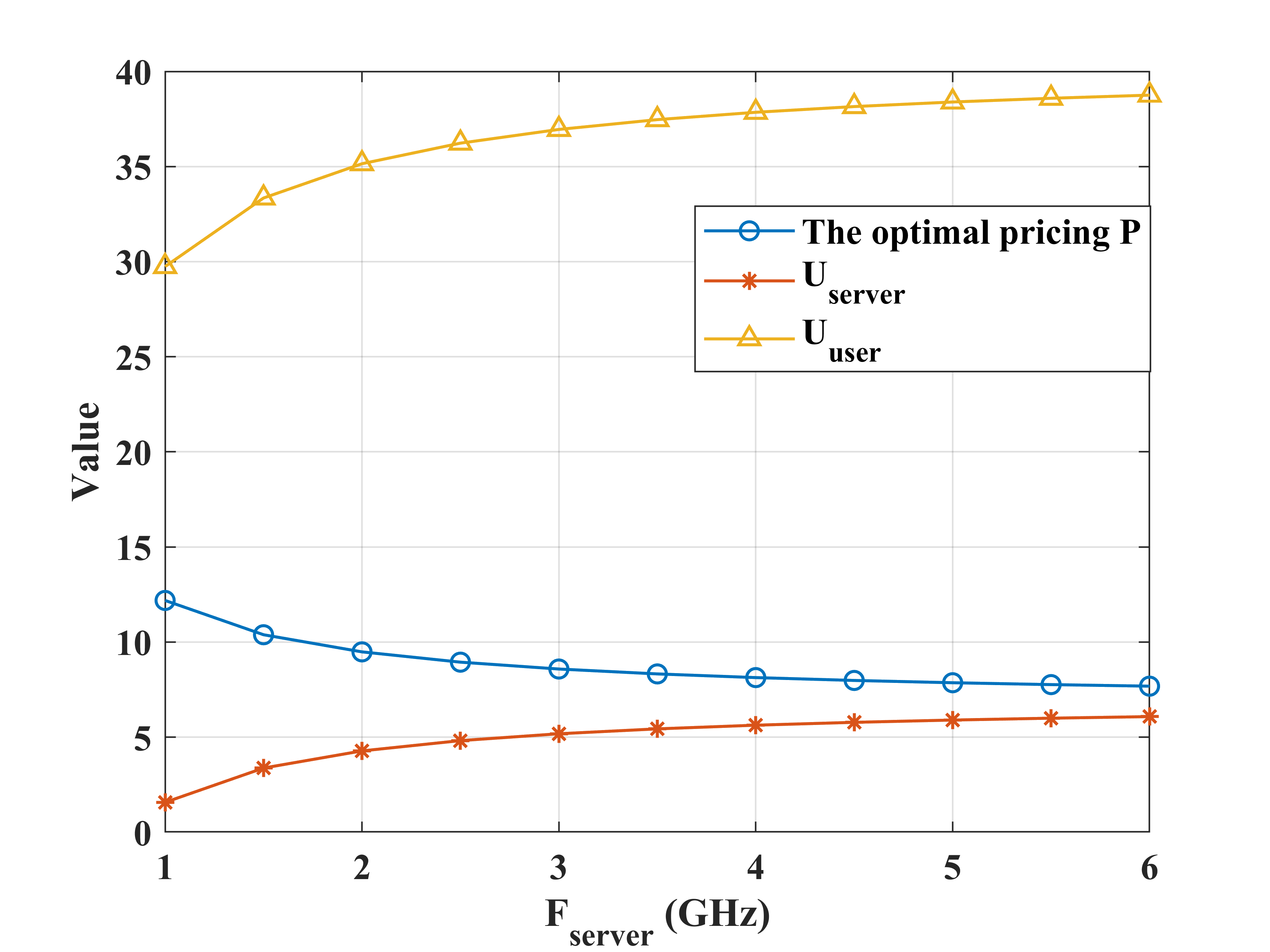}
\caption{Effect of $F_{server}$ on $P$, $U_{user}$, and $U_{server}$. }
\label{fig1}
\end{figure}

From Fig. \ref{fig1}, it is observed that as $F_{server}$ increases, $P$ decreases and then becomes stable, while  $U_{user}$ as well as $U_{server}$ increase and become stable. Firstly, Eq.(\ref{Payplus}) reflects the inverse relationship between $P$ and $F_{server}$. As a result, the increase in $F_{server}$ causes $P$ to fall and flatten out. This also verifies the effectiveness of the proposed pricing model, which can motivate EUs to offload tasks more effectively. Secondly, based on the analysis of $P$, we analyze $U_{user}$ and $U_{server}$ as follows: for $U_{user}$, the increase of $F_{server}$ accelerates the data processing speed. That is, the increase of $T_{save}$ and the decrease of $P$ lead to the increase of $U_{user}$ according to Eq.(\ref{userutility}). In addition, Eq.(\ref{Tsave}) reflects the inverse relationship between $F_{server}$ and $T_{save}$, which leads to the flattening of $U_{user}$, as directly reflected in Eq.(\ref{userutility0}). Similarly, Eq.(\ref{Toffload}) indicates that $F_{server}$ and $T_{offload}$ are inversely proportional, which implies that an increase in $F_{server}$ reduces $T_{offload}$. Moreover, Eq.(\ref{Userver1}) demonstrates that an increase in $F_{server}$ results in a corresponding increase in $U_{server}$.
% \begin{figure}[th]
% \captionsetup{singlelinecheck = false, justification=justified}
% \centering
% \includegraphics[width=3in]{pic/fig2.png}
% \caption{Effect of $F_{server}$ on $P$, $U_{user}$, and $U_{server}$. }
% \label{fig1}
% \end{figure}

\begin{figure}[th]
\captionsetup{singlelinecheck = false, justification=justified}
\centering
\includegraphics[width=3in]{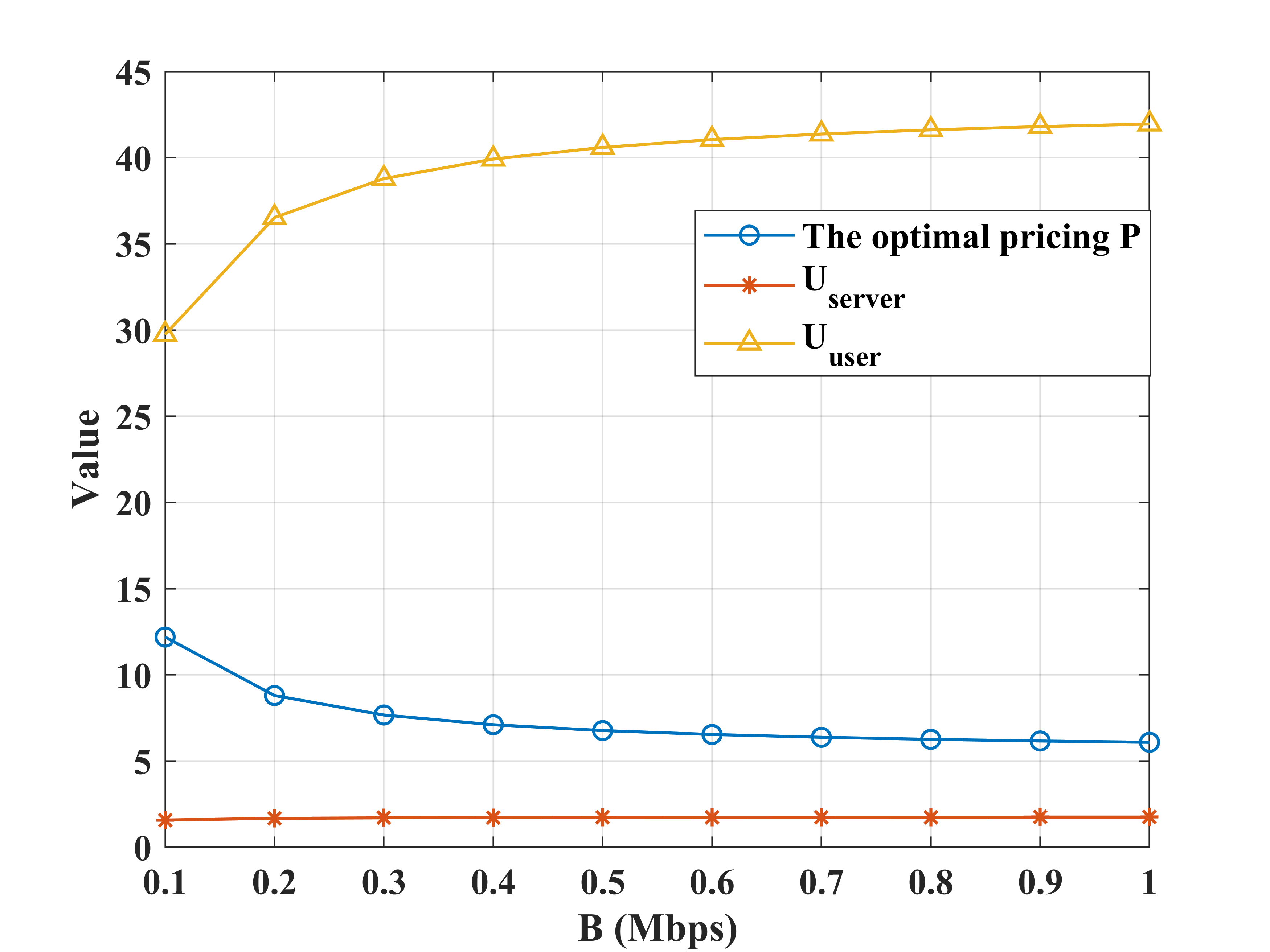}
\caption{Effect of $B$ on $P$, $U_{user}$, and $U_{server}$.}
\label{fig3}
\end{figure}

\subsection{The effect of $B$ on $P$, $U_{user}$, and $U_{server}$.}
As exhibited in Fig. \ref{fig3}, we explore $P$,  $U_{user}$, and $U_{server}$ with diverse $B$ purchased by EUs. We set the CPU frequency resource for the EU to buy $F_{server} = 1GHz$, as well as the amount of offloaded data $500KB$, and we let $B$ vary from $0.1Mbps$ to $1Mbps$.
% \begin{figure}[th]
% \captionsetup{singlelinecheck = false, justification=justified}
% \centering
% \includegraphics[width=3in]{pic/fig3.png}
% \caption{Effect of $B$ on $P$, $U_{user}$, and $U_{server}$.}
% \label{fig3}
% \end{figure}

% By observing Eq.(\ref{userutility0}) and Eq.(\ref{Payplus}), \textcolor{green}{we can see that the analysis of changes of $B$ is similar to that of changes of $F$ when analyzing the effect of changes of $B$. Therefore, the analysis of the effect of $B$ on $P$ and $U_{user}$ is illustrated} in \textbf{subsection B of Section IV}.}

From Fig. \ref{fig3}, it can be seen that as $B$ increases, $P$ decreases and becomes stable, $U_{user}$ increases and becomes stable while $U_{server}$ keeps intact. By examining  Eq.(\ref{userutility0}) and Eq.(\ref{Payplus}), we observe that the effect of $B$ on $P$ and $U_{user}$ is analogous to that of $F$. As this impact has already been analyzed in Subsection B of Section IV, the same reasons also apply to $B$. However, for $U_{server}$, although the increase of $B$ decreases $T_{offload}$, $P$ also decreases accordingly. Thus, there is essentially no discernible change. In addition, by observing Eq.(\ref{Userver1}), we can see that part of the expression related to $B$ is extracted as follows.
\begin{equation}
    B_{part}=\frac{w_1P_u+w_2-1}{log_2(1+(\frac{S}{N})_{uplink})}+ \frac{w_1P_d\alpha+w_2 \alpha-\alpha}{log_2(1+(\frac{S}{N})_{downlink})}.
    \label{Bpart}
\end{equation}

By substituting the data in Table \ref{tabel_3}, we calculate $B_{part}$ $\approx$ $-0.0676$, which represents a relatively small quantity. Therefore, no matter how $B$ changes (i.e., in Eq.(\ref{Userver1}), $B$ is the denominator while $B_{part}$ is the numerator), the overall change is inconspicuous, which further verifies the stable phenomenon of $U_{server}$.

\subsection{The effect of $q$ on $P$, $U_{user}$, and $U_{server}$}
%这意味着服务器鼓励用户多多上传数据，加快整个模型中用户任务的处理速度
As depicted in Fig. \ref{fig4}, we explore $P$,  $U_{user}$, and  $U_{server}$ with different $q$ offloaded by EUs. We set the CPU frequency resource for the EU to buy as $F_{server} = 6GHz$, as well as the bandwidth resource as $B = 1Mbps$, and we let $q$ vary from $100KB$ to $500KB$.
\begin{figure}[!th]
\captionsetup{singlelinecheck = false, justification=justified}
\centering
\includegraphics[width=3in]{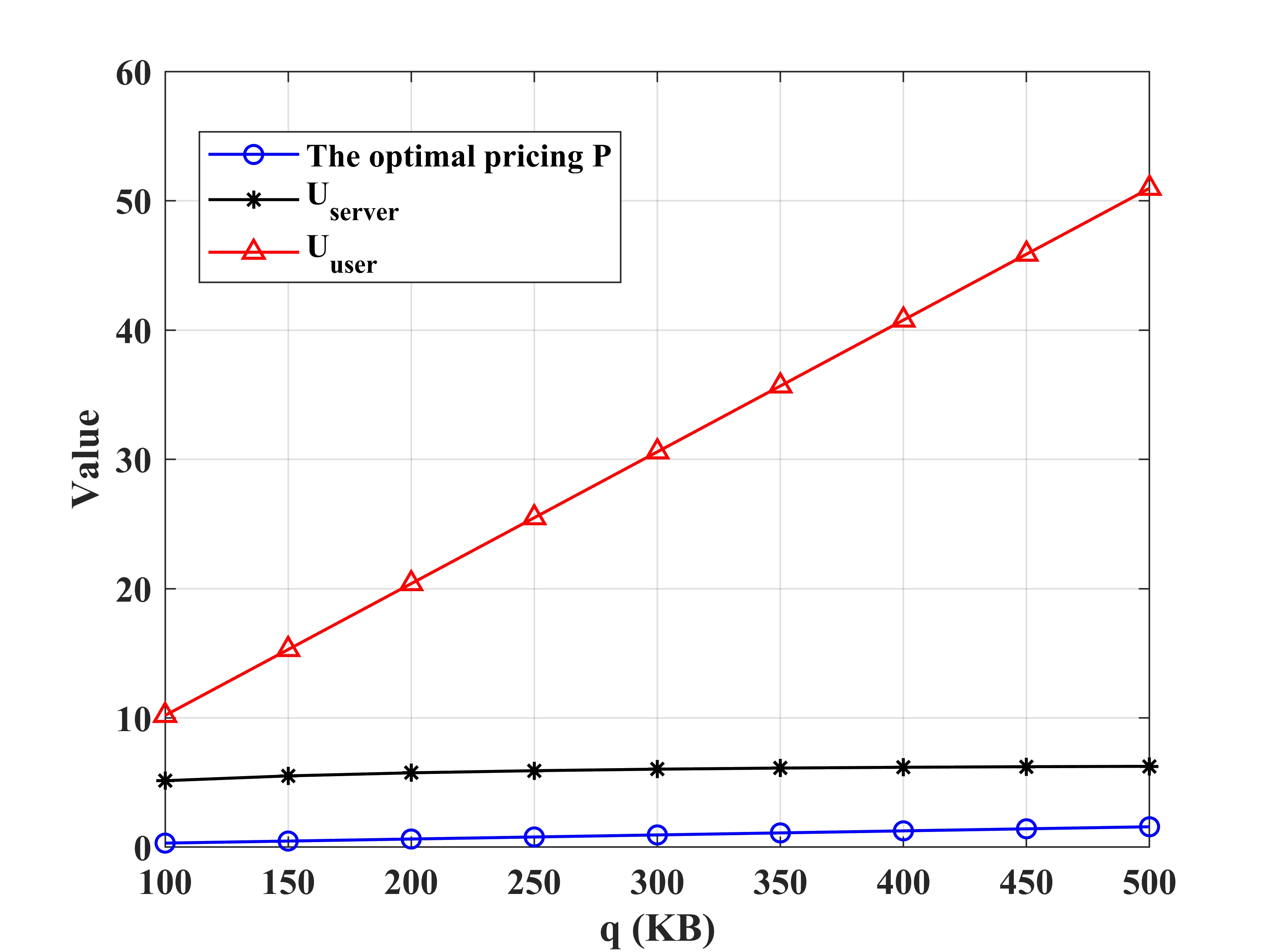}
\caption{Effect of $q$ on $P$, $U_{user}$, and $U_{server}$.}
\label{fig4}
\end{figure}

As $q$ increases, the $U_{user}$ increases while $U_{server}$ and $P$ do not change markedly. For $U_{server}$, the growth rate from the point (100, 5.14252) to (500, 6.25469) is approximately 0.00278, while for $P$, the growth rate from (100, 0.315874) to (500, 1.57937) is approximately 0.00316. With respect to $U_{user}$, when the offloaded data amount gets larger, the $E_{save}$ and $T_{save}$ generated by ES become more obvious due to the constraint and miniature $F_{local}$. Therefore, the $U_{user}$ becomes large, and this trend is intuitively reflected by Eq.(\ref{userutility0}). Similar to the analysis of Eq.(\ref{Bpart}), we extract the expressions related to $q$ in $P$ as follows.
\begin{equation}
    P_{part}=\frac{w_2c}{F_{server}}+\frac{\mathcal{Y}}{B}.
\end{equation}

By substituting the values in Table \ref{tabel_3} into the expression, we can calculate that $P_{part}$ $\in$ [$3.227 \times 10^{-7}$, $2.347 \times 10^{-6}$], which is a small number, which explains $P$ does not change much. Therefore, according to Eq.(\ref{Userver}), the expression affecting $U_{server}$ is illustrated as follows.
\begin{equation}
\begin{aligned}
        U_{affect}&=-T_{offload}+W \\
    &-q\cdot (\frac{1}{R_u}+\frac{c}{F_{server}}+\frac{\alpha }{R_d})+\mu log_2(1+q).
\end{aligned}
\end{equation}

Similarly, we substitute the values in Table \ref{tabel_3} into the expression. That is,  as $q$ changes, $U_{affect}$ does not change much. Therefore, $U_{server}$ maintains a steady level.

\subsection{The effect of $F_{local}$ on $P$, $U_{user}$, and $U_{server}$.}
As shown in Fig. \ref{fig5}, we explore $P$,  $U_{user}$, and  $U_{server}$ under different $F_{local}$ of EUs. We set the CPU frequency resource for the EU to buy as $F_{server} = 6GHz$, as well as the bandwidth resource as  $B = 1Mbps$, and let $q = 500KB$.
\begin{figure}[!th]
\captionsetup{singlelinecheck = false, justification=justified}
\centering
\includegraphics[width=3in]{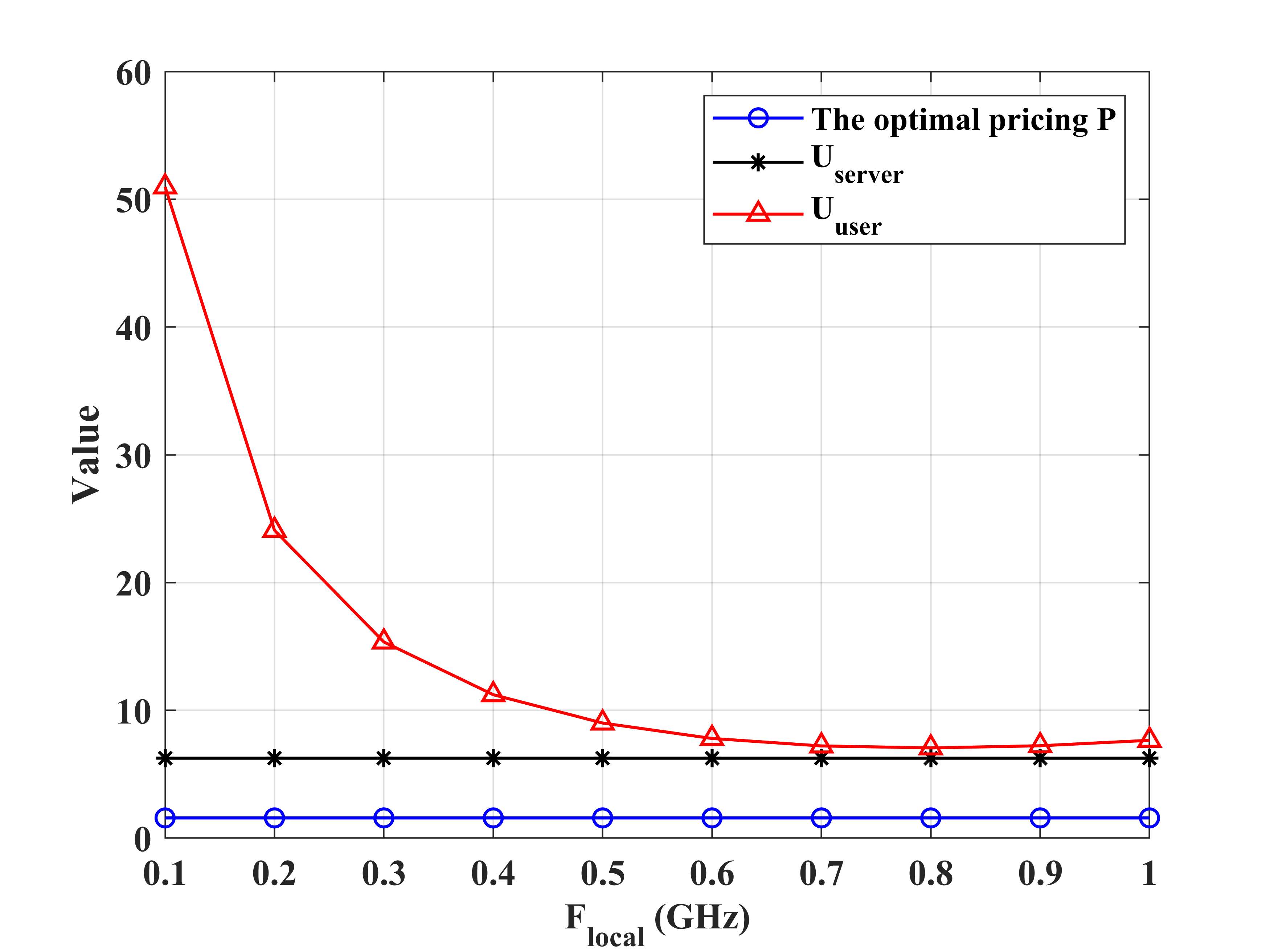}
\caption{Effect of $F_{local}$ on $P$, $U_{user}$, and $U_{server}$.}
\label{fig5}
\end{figure}

As $F_{local}$ increases, $U_{user}$ begins to decrease while $U_{server}$ and $P$ do not change markedly. When $F_{local}$ increases, the EU is more inclined to execute the task locally, resulting in a decrease in $T_{save}$, which leads to a decrease in the $U_{user}$. In addition, by observing Eq.(\ref{userutility0}), we can get the following expression. 
\begin{equation}
    U_{user} \propto \mathcal{X} =\mathcal{M}\cdot F_{local}^2 +\mathcal{N} \cdot \frac{1}{F_{local}},
    \label{relate}
\end{equation}
where $\mathcal{M} = w_1kc$, and $\mathcal{N} = w_2c$. Thus, the analysis of Eq.(\ref{relate}) also leads to the conclusion that $U_{user}$ first decreases and then slightly increases.
Since $P$ and $U_{server}$ are not related to $F_{local}$, $P$, and $U_{server}$ remain unchanged.

\subsection{The effect of $F_{server}$ and $B$ on $U_{user}$, $P$, and $U_{server}$.}

\begin{figure*}[!htbp]
\captionsetup{singlelinecheck = false, justification=justified}
\centering
\subfigure[Effect of $F_{server}$ and $B$ on  $U_{user}$.]{
\begin{minipage}[t]{0.33\linewidth}
\centering
\includegraphics[width=2.56in]{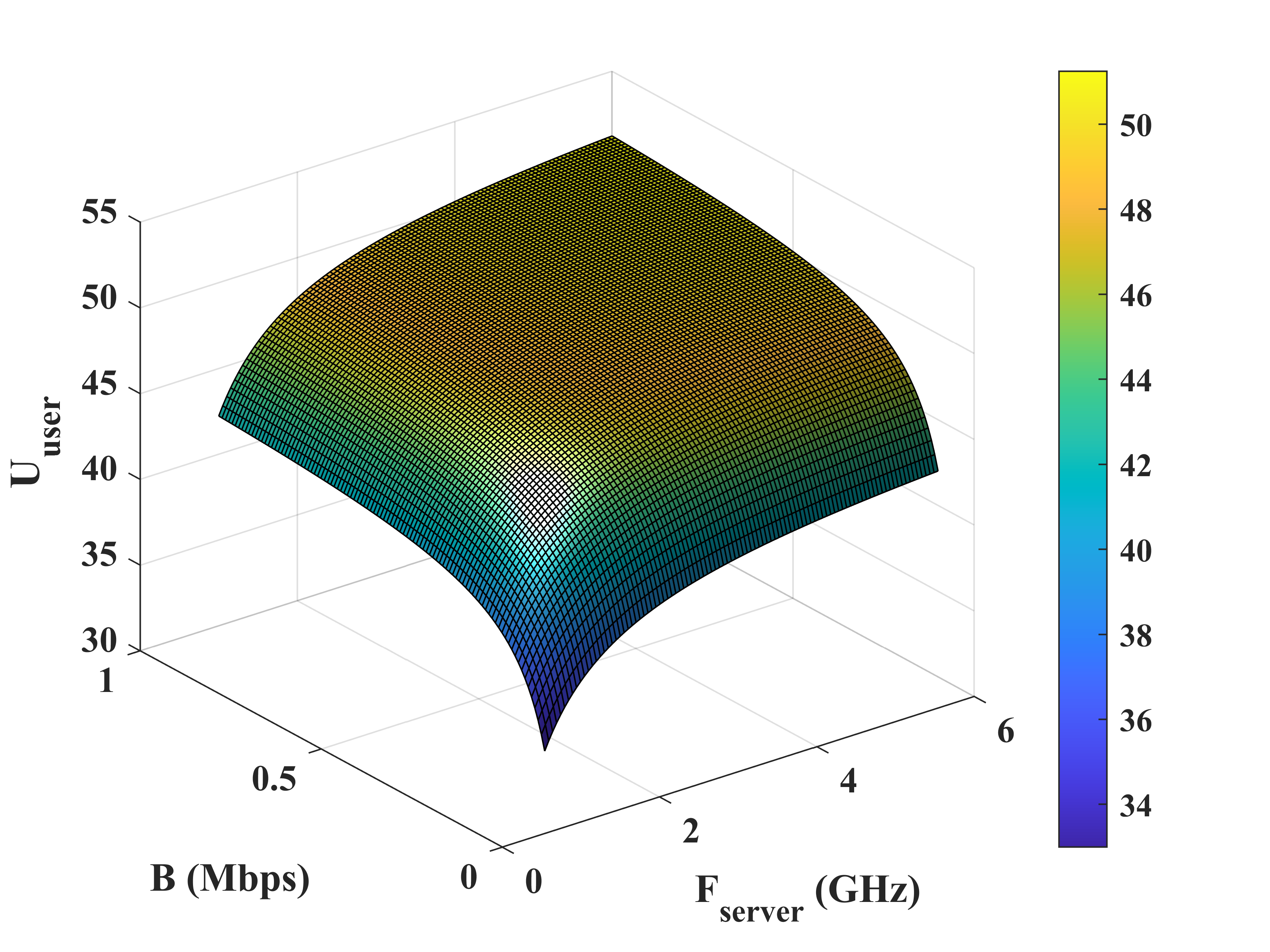}
\label{fig_6a}
\end{minipage}
}%
\subfigure[Effect of $F_{server}$ and $B$ on $P$.]{
\begin{minipage}[t]{0.33\linewidth}
\centering
\includegraphics[width=2.56in]{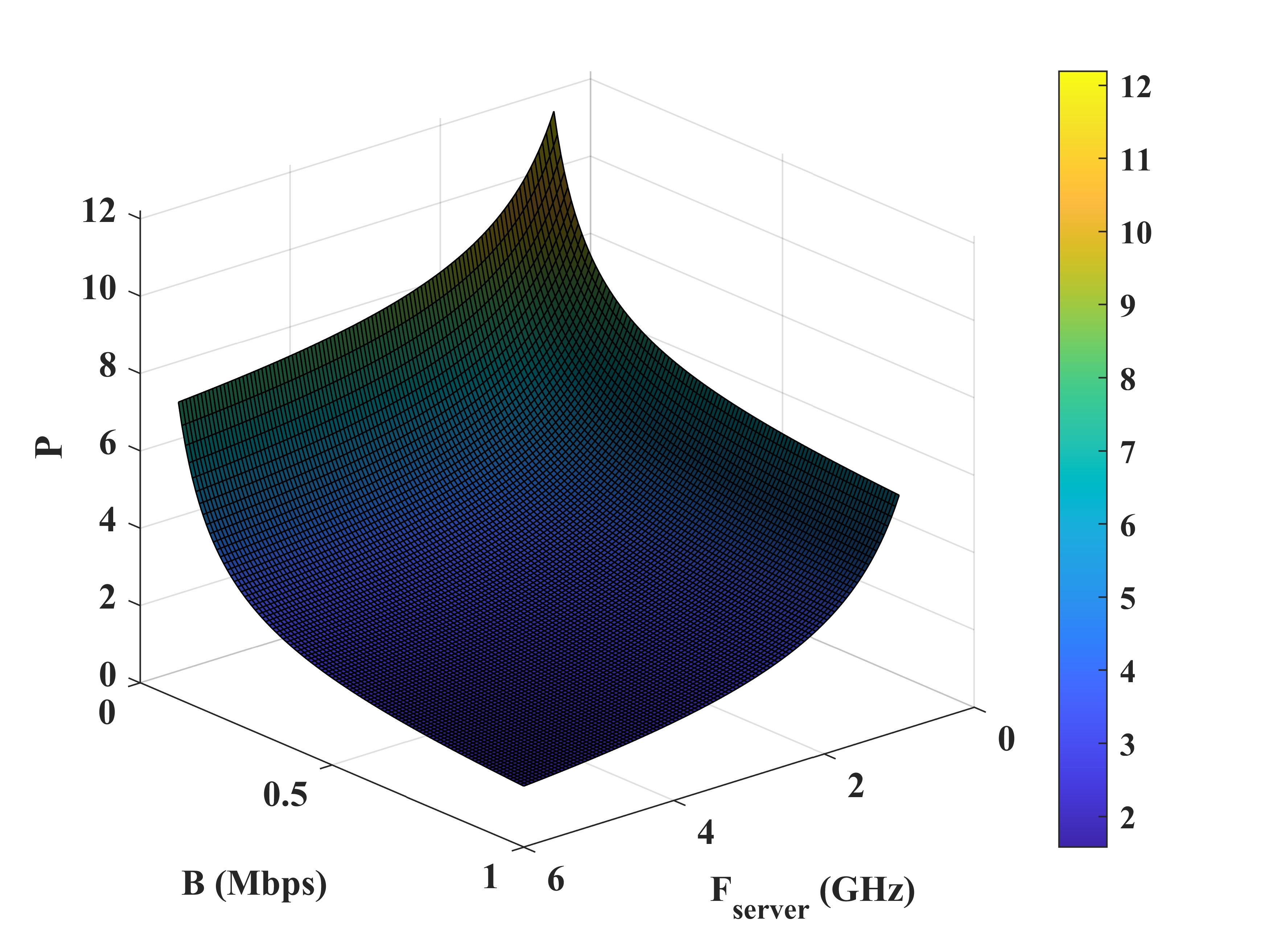}
\label{fig_6b}
\end{minipage}
}%
\subfigure[Effect of $F_{server}$ and $B$ on $U_{server}$.]{
\begin{minipage}[t]{0.33\linewidth}
\centering
\includegraphics[width=2.56in]{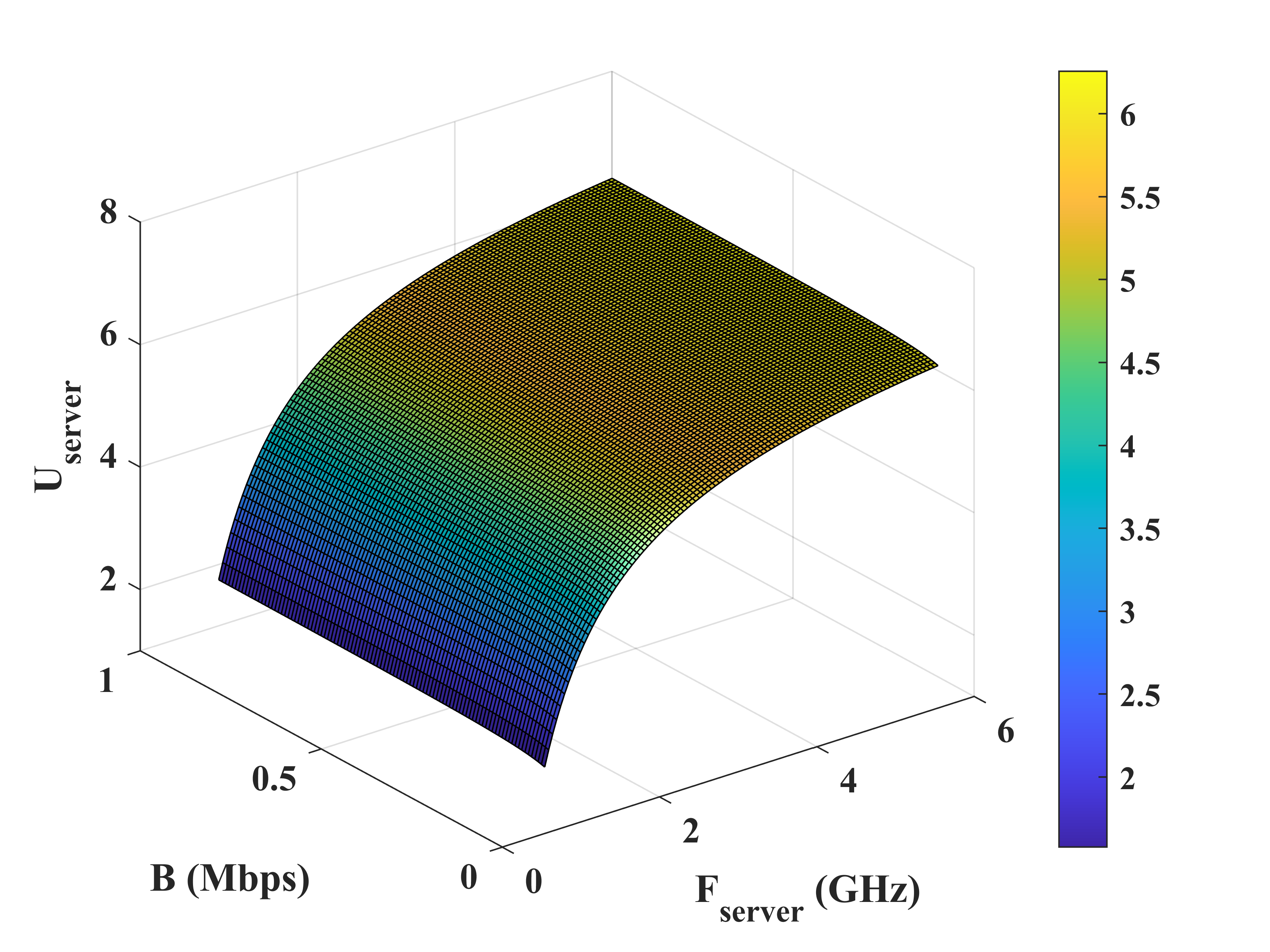}
\label{fig_6c}
\end{minipage}
}%
\centering
\caption{Effect of $F_{server}$ and $B$ on $U_{user}$, $P$, and $U_{server}$.}
\label{Fig6}
\end{figure*}

In addition,  we also analyze the  $U_{user}$, $P$, and $U_{server}$ for different $F_{server}$ and $B$, as shown in Fig. \ref{Fig6}. We set the bandwidth resource purchased by the EU $B \in [0.1Mbps, 1Mbps]$, the CPU frequency $F \in [1GHz, 6GHz]$, and the amount of offloaded data $q$ as $500KB$. 

From Fig. \ref{fig_6a}, it is observed that the $U_{user}$ arrives at a local maximum at the point $(1Mbps, 6GHz)$, which is also confirmed by Eq.(\ref{criticalpoint}). In addition, with the increase of resource allocation, $U_{user}$ gradually becomes stable. This prevents EUs from excessively seizing ES resources to boost $U_{user}$. As a result, ES resources are allocated and utilized more reasonably.
% \begin{figure}[!th]
% \captionsetup{singlelinecheck = false, justification=justified}
% \centering
% \includegraphics[width=3in]{fig2.png}
% \caption{Effect of different $F_{server}$ and $B$ on  $U_{user}$. }
% \label{fig2}
% \end{figure}
% As shown in Fig. \ref{fig_6b}, $P$ gradually decreases as resource allocation increases. This conclusion is also reached by combining Fig. \ref{fig1}  and Fig. \ref{fig3}. The proposed} pricing model encourages EUs to offload tasks, \textcolor{green}{verifies the effectiveness of the pricing scheme, and enhances the overall system performance.}
As shown in Fig. \ref{fig_6b}, $P$ gradually decreases as resource allocation increases. This observation is consistent with the trends illustrated in Fig. \ref{fig1} and Fig. \ref{fig3}. The proposed pricing model effectively incentivizes EUs to offload tasks by reducing prices as more resources are purchased. This not only validates the effectiveness of the pricing scheme but also enhances overall system performance, since it encourages EUs to acquire additional computational resources, accelerating task execution and enhancing system efficiency.

% \begin{figure}[!th]
% \captionsetup{singlelinecheck = false, justification=justified}
% \centering
% \includegraphics[width=3in]{fig6.png}
% \caption{Effect of different $F_{server}$ and $B$ on $P$.}
% \label{fig6}
% \end{figure}
Besides, as we can see from Fig. \ref{fig_6c}, the $U_{server}$ tends to stabilize as the resources allocated to the EU reach a certain level, which is typically achieved when $F_{server}$ exceeds $5GHz$. This limits the ES's ability to increase its $U_{server}$ by over-selling resources to a single affluent EU, allowing the ES to allocate resources to EUs appropriately.
% \begin{figure}[!t]
% \captionsetup{singlelinecheck = false, justification=justified}
% \centering
% \includegraphics[width=3in]{fig7.png}
% \caption{Effect of different $F_{server}$ and $B$ on $U_{server}$.}
% \label{fig7}
% \end{figure}

\subsection{Near-optimal resource allocation with DISC-PSO algorithm}
Next, we conducted 50 simulation runs to compare the $U_{user}$ of seeking near-optimal value with different algorithms (i.e., as shown in Fig. \ref{fig8}), the corresponding number of iterations (i.e., as shown in Fig. \ref{fig9}), and the corresponding resource distribution (i.e., as shown in Fig. \ref{Fig10}). In addition, we have carried out statistics on the physical quantities as shown in Table \ref{tabel_4}.

From Fig. \ref{fig8}, it can be seen that the $U_{user}$ of the DISC-PSO is always at a high value and relatively stable compared with other algorithms, and this is also reflected in Table \ref{tabel_4}.
In addition, DISC-PSO has the lowest standard deviation on $U_{user}$, which further verifies the stability of the proposed algorithm.
% In addition, the smaller standard deviation means the stability of the algorithm.

Fig. \ref{fig9} shows the number of iterations required to iterate to a qualified $U_{user}$, and we set the maximum number of iterations to 50. It is observed that DISC-PSO satisfies the condition of Eq.(\ref{condition}) in a few iterations. Compared with PSO, GA, and DE, the number of iterations decreased by 73.9\%, 96.5\%, and 83.2\% respectively as shown in Table \ref{tabel_4}. This demonstrates that our proposed algorithm is highly efficient, achieving faster convergence with significantly fewer iterations.
\begin{figure}[!t]
\captionsetup{singlelinecheck = false, justification=justified}
\centering
\includegraphics[width=3in]{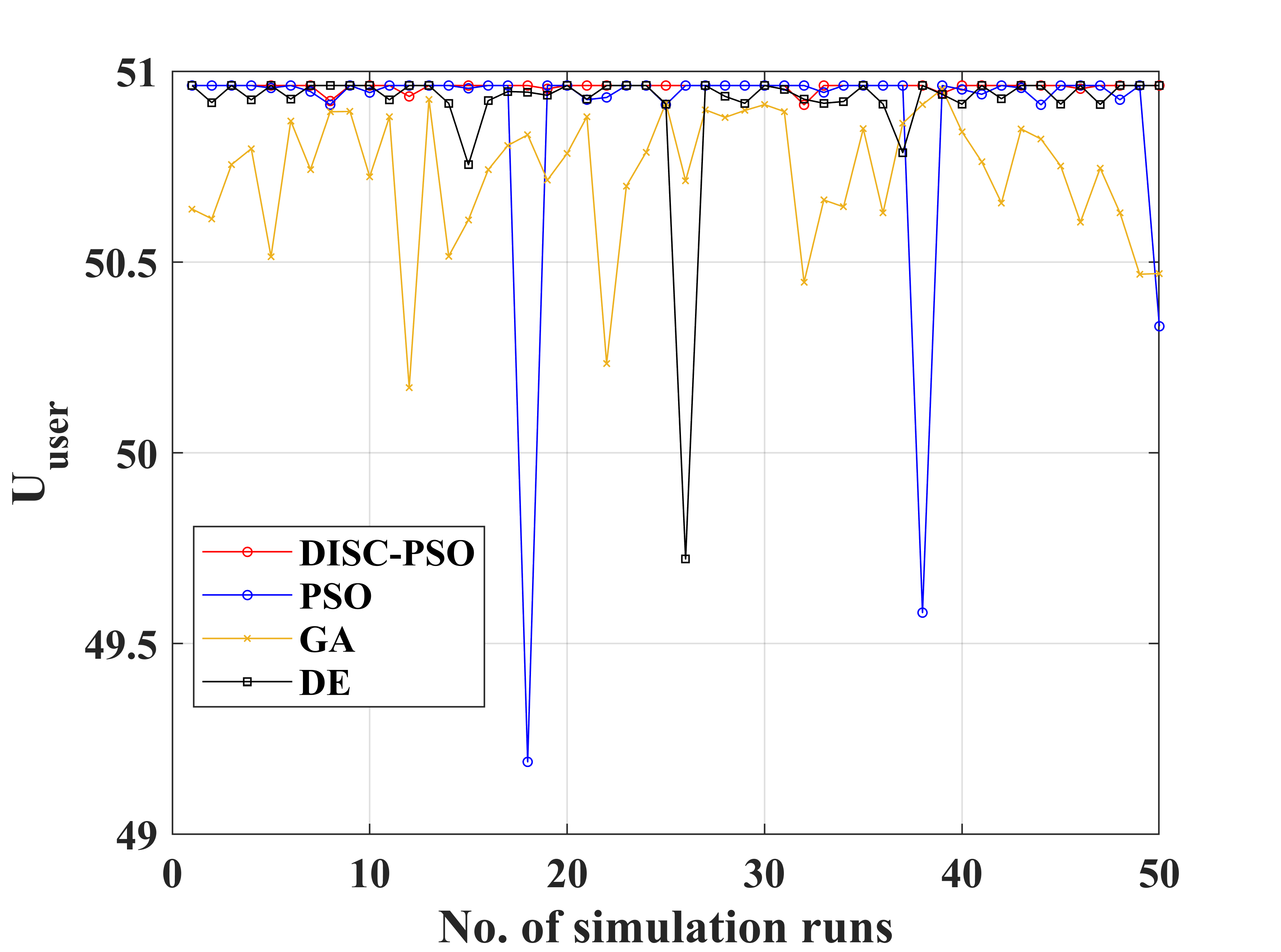}
\caption{$P$ under different algorithms. }
\label{fig8}
\end{figure}

\begin{figure}[!t]
\captionsetup{singlelinecheck = false, justification=justified}
\centering
\includegraphics[width=3in]{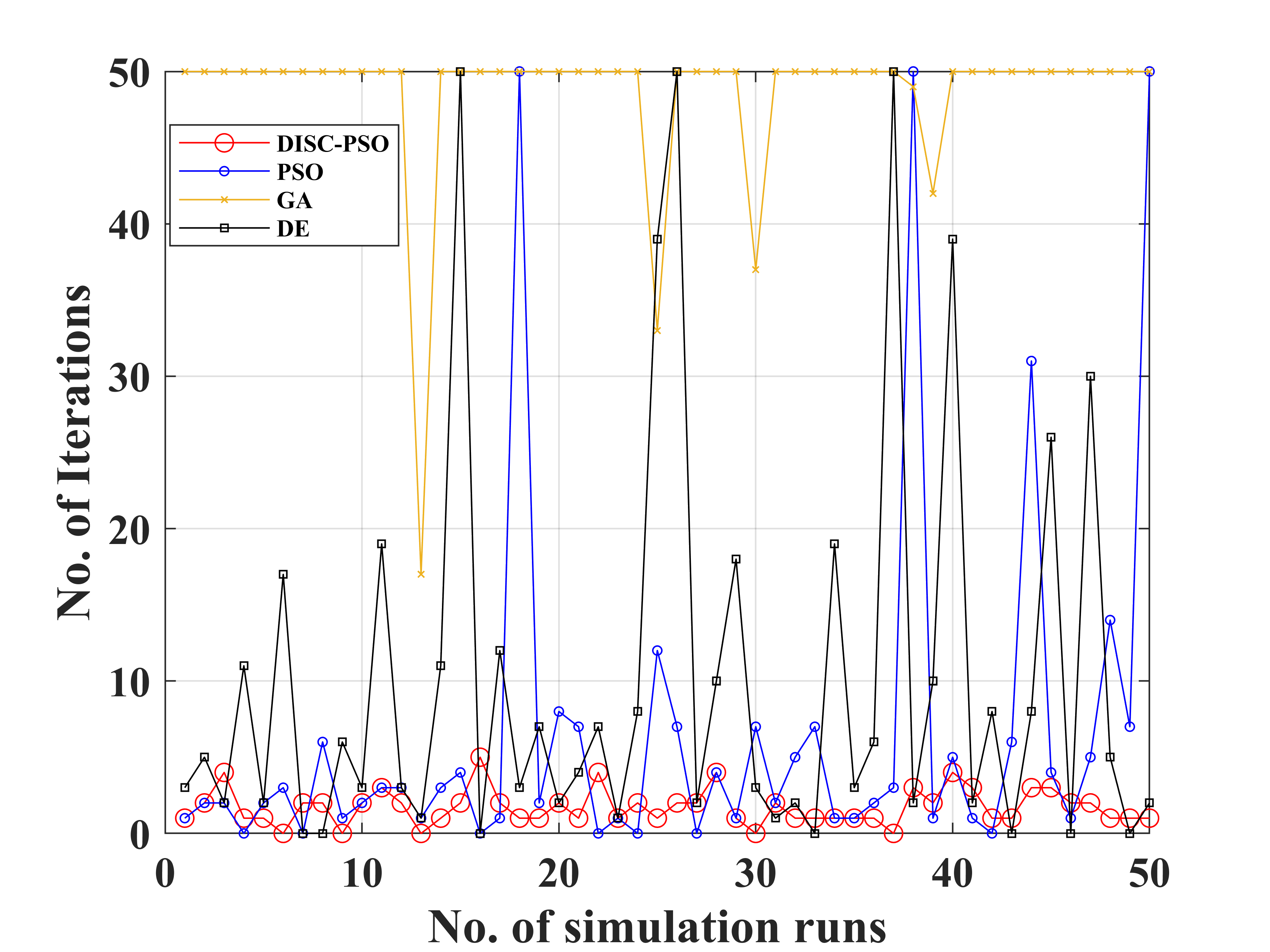}
\caption{The number of iterations under different algorithms. }
\label{fig9}
\end{figure}

As depicted in Fig. \ref{Fig10}, we counted the convergence point distribution positions obtained by different algorithms after 50 simulation runs. Compared with other baseline algorithms, the point distribution of DISC-PSO is more concentrated and less. This means that the DISC-PSO can achieve relatively stable resource allocation and $U_{user}$, which helps the ES more accurately to predict the amount of resources that should be allocated to the EU, thereby ensuring rationality and fairness in resource allocation.
\begin{figure*}[!htbp]
\centering
\captionsetup{singlelinecheck = false, justification=justified}
% \label{Fig_2}
\subfigure[DISC-PSO.]{
\begin{minipage}[t]{0.5\linewidth}
\centering
\includegraphics[width=3in]{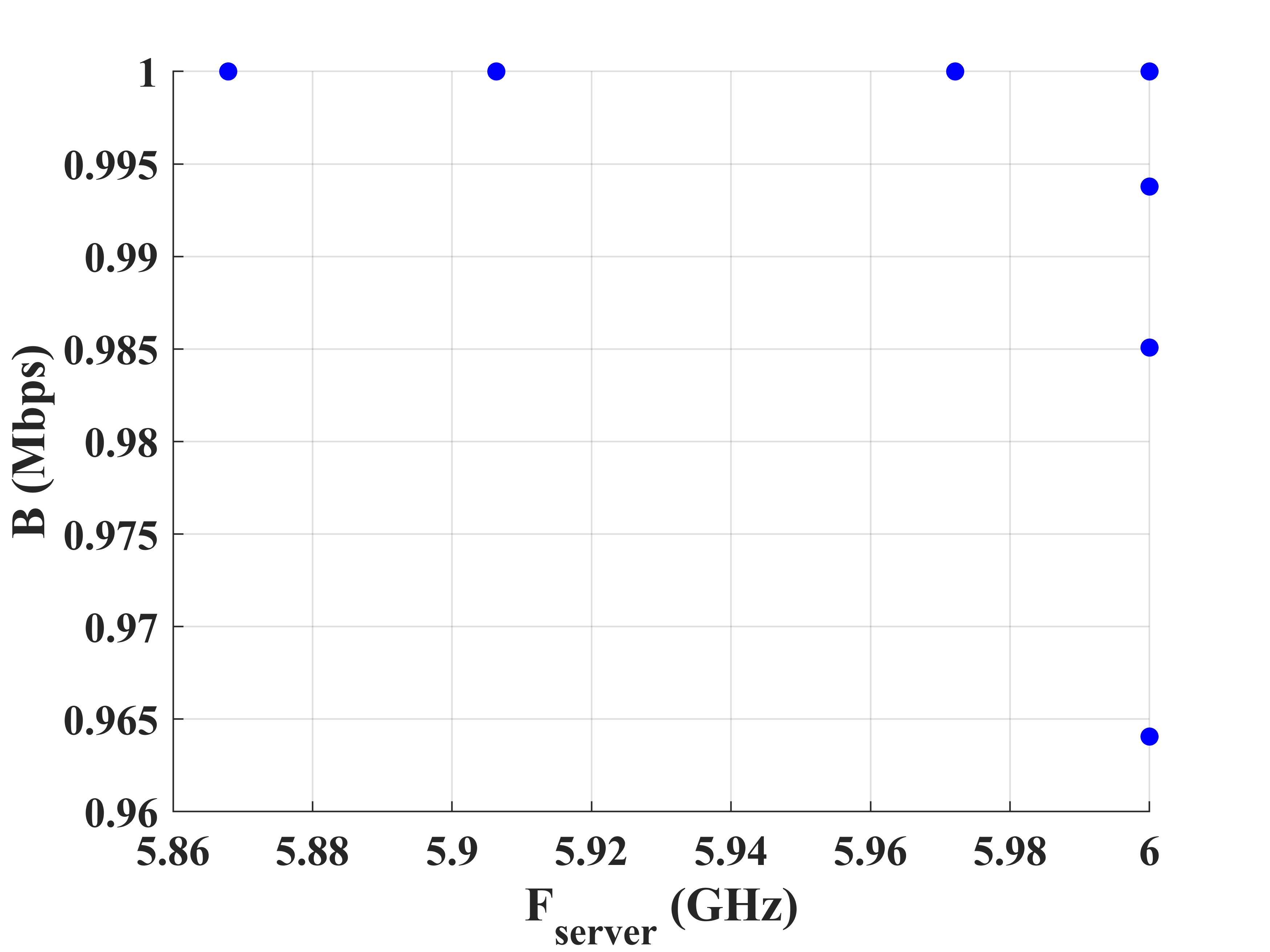}
\label{IPSO}
%\caption{fig1}
\end{minipage}%
}%
\subfigure[PSO.]{
\begin{minipage}[t]{0.5\linewidth}
\centering
\includegraphics[width=3in]{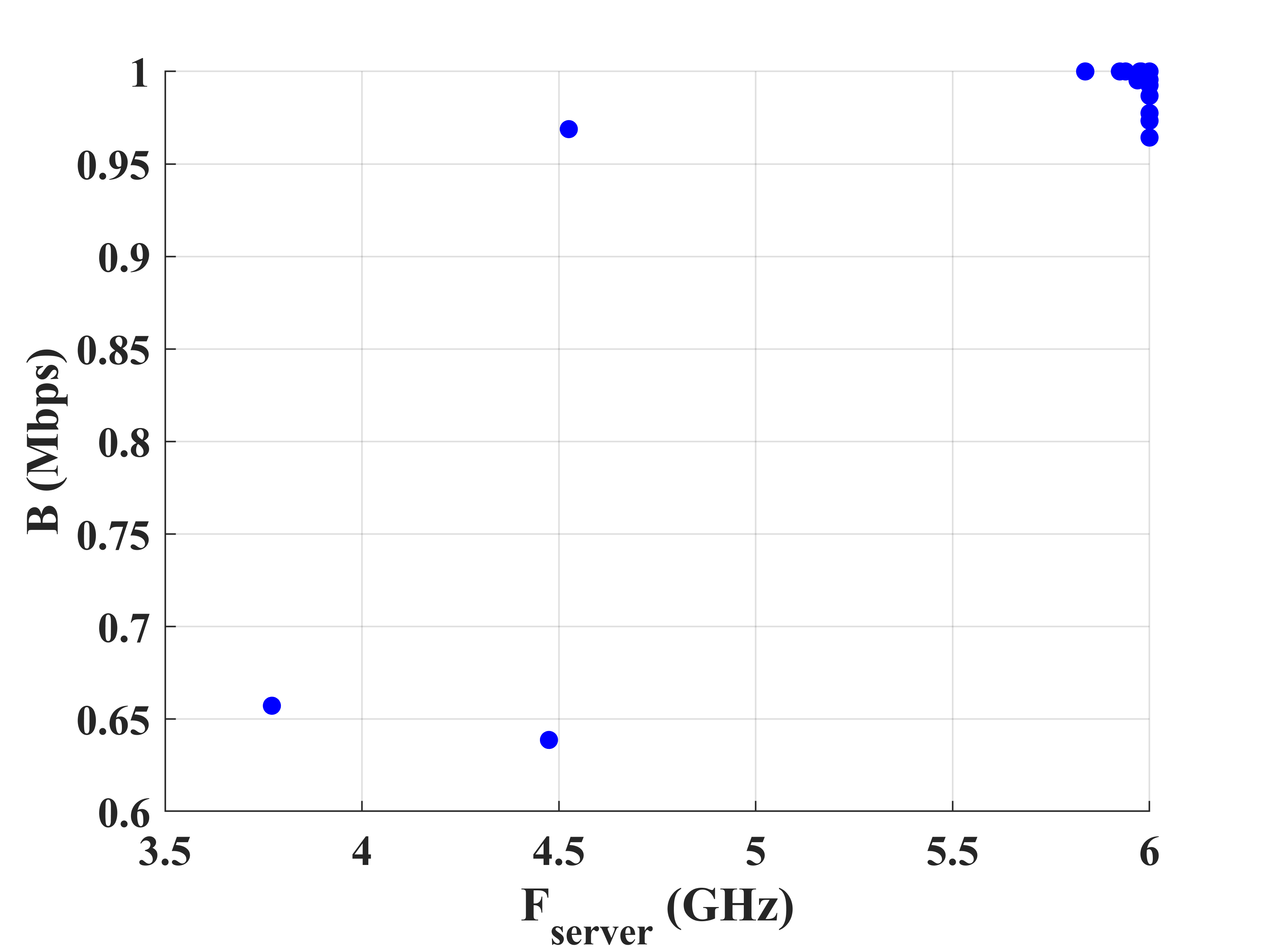}
\label{PSO}
%\caption{fig2}
\end{minipage}%
}%
\hfill
\subfigure[GA.]{
\begin{minipage}[t]{0.5\linewidth}
\centering
\includegraphics[width=3in]{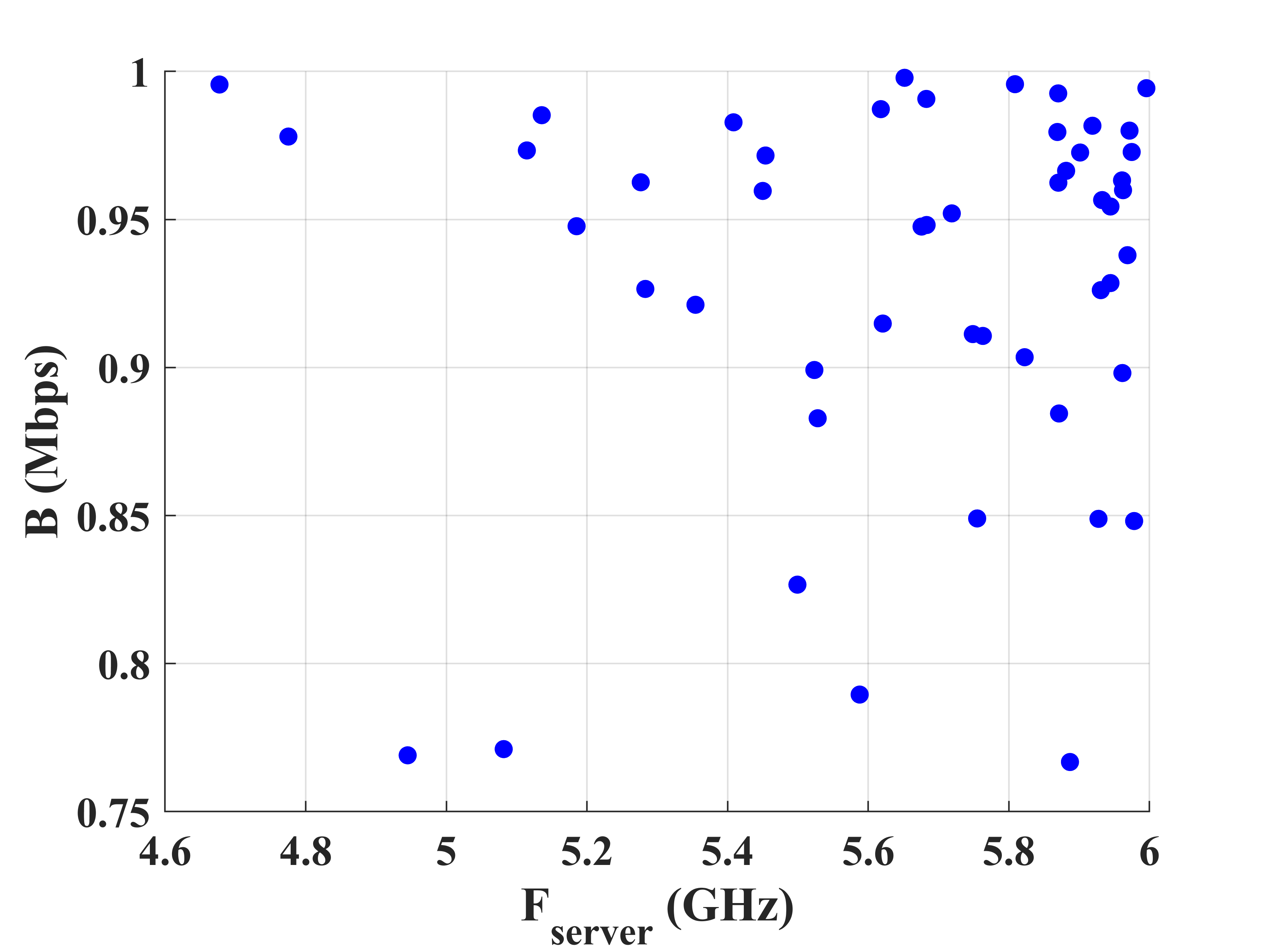}
\label{GA}
%\caption{fig2}
\end{minipage}
}%
\subfigure[DE.]{
\begin{minipage}[t]{0.5\linewidth}
\centering
\includegraphics[width=3in]{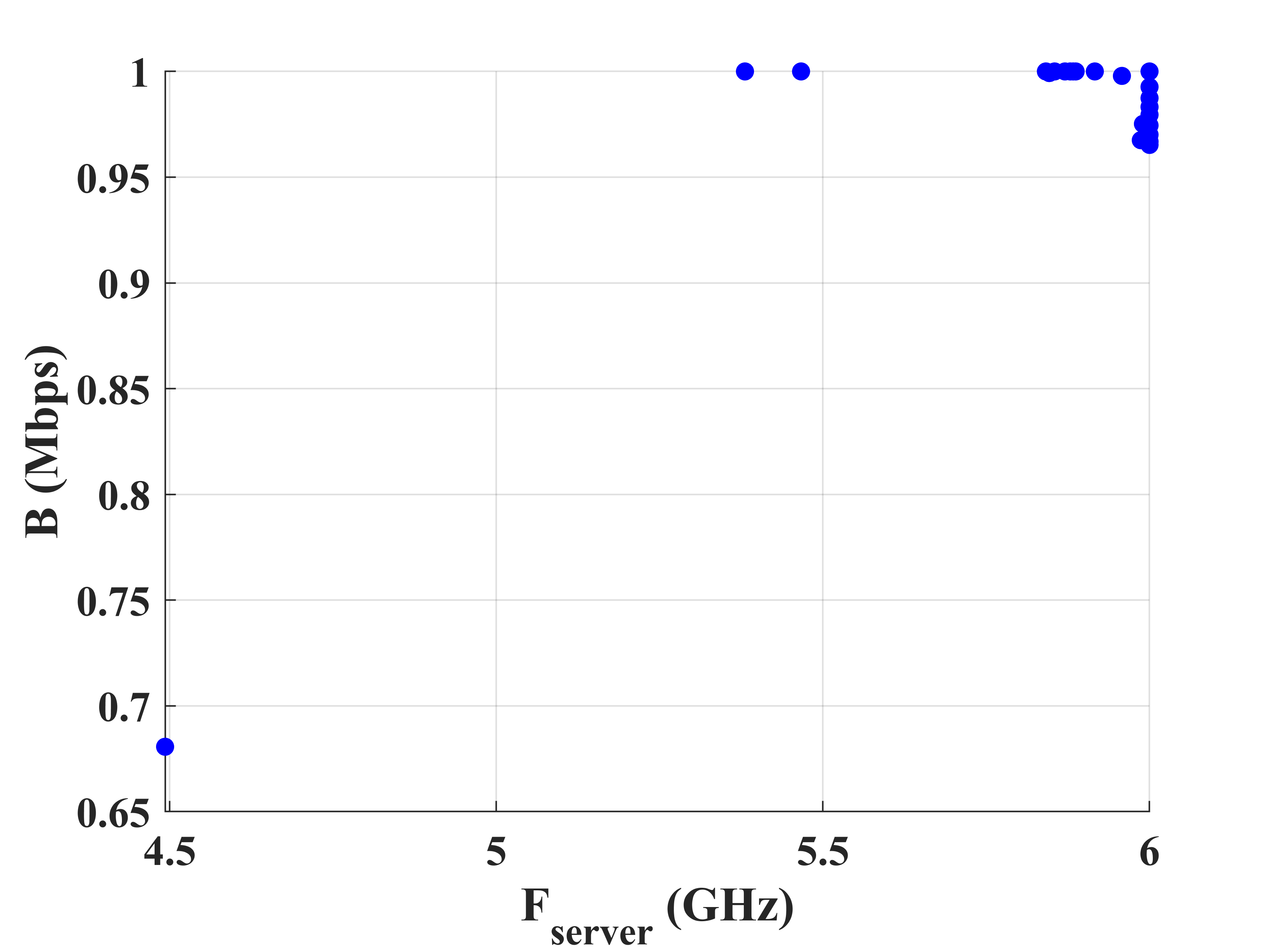}
\label{DE}
%\caption{fig2}
\end{minipage}
}%
\centering
\caption{Distribution of near-optimal points under different algorithms.\protect \footnotemark}
\label{Fig10}
\end{figure*}

\begin{table}[h]
    \centering
    \caption{SUMMARY OF SIMULATION RESULTS OF DIFFERENT ALGORITHMS}
    \begin{tabular}{|c|c|c|c|}
         \hline  
        \textbf{Algorithm}  &  \textbf{Average $U_{user}$}  & \textbf{Std} & \textbf{Average No. of iterations}  \\ 
        \hline  
        DISC-PSO & 50.96 & 0.01019 & 1.72 \\ 
        \hline 
        PSO & 50.88 & 0.3232 & 6.58 \\
        \hline
        GA & 50.73 & 0.1732 &48.56  \\
        \hline
        DE & 50.91 & 0.1765 &10.24  \\
        \hline        
    \end{tabular}
    \label{tabel_4}
\end{table}
\footnotetext[4]{Only the unrepeated points are drawn.}

\section{conclusion}
% This paper proposes a dynamic pricing scheme for ES resources and validates the effectiveness of the proposed pricing scheme through extensive simulation results. Subsequently, utility functions for EUs and ESs are derived based on the pricing function. It is then demonstrated that the EU utility function has the local maximization within the search range, from which optimal resource allocation is obtained. Additionally, it is observed that there is minimal change in utility around the maximum as resource allocation increases. To achieve more reasonable resource allocation, we introduce the IPSO algorithm to find a suboptimal resource allocation. Furthermore, we validate the effectiveness of this algorithm by considering factors such as the number of iterations, average utility function, standard deviation of utility function, and resource distribution.}
In this paper, we develop a dynamic pricing based near-optimal resource allocation scheme and validate its effectiveness through extensive simulation results. Specifically, we derive utility functions for both EUs and ESs based on the proposed pricing scheme. Our theoretical analysis demonstrates that the EU utility function reaches a local maximum within the search range, which facilitates optimal resource allocation. Moreover, we find that near the local maximum, more resource allocation results in only modest improvements in utility. Therefore, to develop a more rational resource allocation scheme, we propose the DISC-PSO algorithm to identify a near-optimal allocation quantity. The effectiveness of the DISC-PSO algorithm is further validated by examining factors such as the number of iterations, average utility, standard deviation of utility, and resource distribution. In our future work, we will explore more dynamic factors and implement the proposed scheme into real-world experiments for practical validation. 

% However, in our simulations, certain parameters were fixed or had limited variability, which may not fully capture the complexities of real-world scenarios. In future work, we aim to incorporate more dynamic factors and obtain data that better reflects practical conditions.
% \textcolor{green}{The proposed pricing strategy offers several key advantages. First, it rationalizes pricing by reducing costs as users purchase more resources, encouraging them to offload more tasks, which in turn accelerates overall system processing and enhances performance. Additionally, the dynamic pricing scheme ensures that both EU and ES utilities reach local optima while also limiting the excessive consumption of resources—since utilities gain diminish as more resources are acquired, effectively preventing overuse. The DISC-PSO algorithm further refines resource allocation by targeting near-optimal utilities around the local maximum. This approach allows for efficient resource savings without significantly impacting the utility of either users or servers. Furthermore, the algorithm is characterized by its fast convergence and high robustness, as demonstrated in our conclusions.}

\begin{appendices}
\section{}
\label{negative}
The Hessian matrix $\mathcal{H}$ for $U_{user}$ is shown below.
\begin{align}
  \mathcal{H} =   
\begin{bmatrix}  
-\frac{2w_2cq}{F_{server}^3} & 0 \\  
0 & -\frac{2q\mathcal{Y}}{B^3} 
\end{bmatrix}.
\end{align}

Since $\mathcal{H}$ is a diagonal matrix, we can get its eigenvalues $\lambda_1$ and $\lambda_2$ as shown below.
\begin{equation}
    \left\{ \begin{array}{l}  
    \lambda_1 = -\frac{2w_2cq}{F_{server}^3}, \\  
    \lambda_2 = -\frac{2q\mathcal{Y}}{B^3}.  
\end{array} \right.
% \label{featurevalue}
\end{equation}

Because \begin{equation}
    \left\{ \begin{array}{l}  
    w_2 >0, \\  
    c > 0,  \\  
    q > 0, \\
    F_{server} > 0,\\
    B > 0,
\end{array} \right.
\end{equation}
$\lambda_1 < 0$. In addition, we know $\mathcal{Y} = \frac{w_1P_u+w_2}{log_2(1+(\frac{S}{N})_{uplink})}+\frac{w_1P_d\alpha+w_2\alpha}{log_2(1+(\frac{S}{N})_{downlink})}$ (i.e., Eq.(\ref{yexpression})). Since \begin{equation}
    \left\{ \begin{array}{l}  
    w_1 >0, \\  
    P_u > 0,  \\  
    P_d > 0, \\
    \alpha > 0,\\
    log_2{(1+\mathcal{D})}>0,
\end{array} \right.
\end{equation}
where $\mathcal{D} > 0$ is the signal to noise ratio, $\mathcal{Y}>0$. Therefore, $\lambda_2 < 0$.

To sum up, since the eigenvalues $\lambda_1$ and $\lambda_2$ are both greater than zero, the $\mathcal{H}$ is a negative definite matrix.
% \begin{align}
%      |\mathcal{H}| &= (-\frac{2w_2cq}{F_{server}^3})\times (-\frac{2q\mathcal{Y}}{B^3}) \nonumber \\
%      & = \frac{4q^2w_2c\mathcal{Y}}{(F_{server}\cdot B)^3}
% \end{align}

\section{}
\label{slow}
Denote $U_{user} = f(F_{server}, B)$. From Eq.(\ref{criticalpoint}), we know that $f(F_{server}, B)$ has a local maximum. We assume the local maximum point is ($F_{server}^0$, $B^0$). In addition, the Hessian matrix $\mathcal{H}$ at this point is negative definite, with the gradient \(\nabla f(F_{server}^0, B^0)\) = 0. The second-order Taylor expansion is expressed as:

\begin{equation}
\begin{aligned}
    f(&F_{server}, B) \approx f(F_{server}^0, B^0) \\ &+ \frac{1}{2}\begin{pmatrix} F_{server} - F_{server}^0 \\ B - B^0 \end{pmatrix}^T\mathcal{H}\begin{pmatrix} F_{server} - F_{server}^0 \\ B - B^0 \end{pmatrix}.
\end{aligned}
\end{equation}

Let the eigenvalues of the Hessian matrix $\mathcal{H}$ as $\lambda_1$ and $\lambda_2$. We use the eigenvectors $v_1$ and $v_2$ to present the corresponding eigenvalues. The vector $z = \begin{pmatrix}
    F_{server} - F_{server}^0 \\ B - B^0
\end{pmatrix}$ can be expressed as a linear combination of the eigenvectors:
\begin{equation}
    z = c_1v_1 + c_2v_2,
\end{equation}
where $c_1$ and $c_2$ are the coefficients representing the vector z in terms of the eigenvector basis $v_1$ and $v_2$.

In this expression, the quadratic form $z^T\mathcal{H}z$ can be simplified as:
\begin{equation}
    z^T\mathcal{H}z = \lambda_1c_1^2 + \lambda_2c_2^2.
\end{equation}

Therefore, the quadratic approximation formula is:
\begin{equation}
    f(F_{server}, B) \approx f(F_{server}^0, B^0)  + \frac{1}{2}(\lambda_1c_1^2 + \lambda_2c_2^2).
\end{equation}

Since $\lambda_1$ and $\lambda_2$ are negative, this expression becomes:
\begin{equation}
    f(F_{server}, B) \approx f(F_{server}^0, B^0)  - \frac{1}{2}(|\lambda_1|c_1^2 + |\lambda_2|c_2^2).
\end{equation}

Here, we obtain the sizes of the eigenvalues $\lambda_1$ and $\lambda_2$ by substituting the simulation parameters in Table \ref{tabel_3}. Thus, $\lambda_1$ and $\lambda_2$ are approximately $10^{-18}$ and $10^{-13}$, respectively.
Since $\lambda_1$ and $\lambda_2$ are very small, that is, $|\lambda_1|$ and $|\lambda_2|$ are very small positive number. This means:
\begin{equation}
    \left|f(F_{server}, B) - f(F_{server}^0, B^0)\right| = \left|- \frac{1}{2}(|\lambda_1|c_1^2 + |\lambda_2|c_2^2)\right|.  
\end{equation}

% then $|\lambda_1|c_1^2 + |\lambda_2|c_2^2$ will be very small when ($F_{server}$, $B$) is close to ($F_{server}^0$, $B^0$)

Since $|\lambda_1|$ and $|\lambda_2|$ are very small, even if $c_1$ and $c_2$ are not particularly small\footnote{The sizes of $c_1$ and $c_2$ determine how far the point ($F_{server}$, $B$) is from the point ($F_{server}^0$, $B^0$). The smaller $c_1$ and $c_2$ are, the closer the two points are.}, $|\lambda_1|c_1^2 + |\lambda_2|c_2^2$ will be very small. Therefore, the overall change in the function value will be still very small.

In summary, since the eigenvalues $\lambda_1$ and $\lambda_2$ of the Hessian matrix are very small, the function value decreases very slowly near the local maximum ($F_{server}^0$, $B^0$). This is because the descent term  $\frac{1}{2}(|\lambda_1|c_1^2 + |\lambda_2|c_2^2)$ in the quadratic approximation formula does not significantly change the function value when $|\lambda_1|$ and $|\lambda_2|$ are very small. Therefore, the rate of descent of the function near the local maximum is slow. When the eigenvalues $|\lambda_1|$ and $|\lambda_2|$ are small, the curvature of the function in these directions is small, so the value of the function drops very slowly near the local maximum.

\section{} 
\label{delte}
We analyzed some of the data in Fig. \ref{fig1}\footnote{We have omitted the analysis of Fig. \ref{fig3} data here. Fig. \ref{fig3} shows that with an increase in $B$, the $U_{server}$ is basically unchanged. Our goal is to omit the $U_{server}$ analysis, which means we only need to consider the $U_{user}$. Therefore, only the $U_{server}$ and $U_{user}$ under different $F$ is analyzed here, and the utility analysis under different $B$ is omitted.}, as shown in Table \ref{tabel_5}. From the table, we can see that $\Delta U_{user}$ is twice as large as $\Delta U_{server}$ under the same conditions. In addition, when Eq. (\ref{condition}) is satisfied, the algorithm stops iteration. Let's make the following formula true.

\begin{align}
    \frac{U_{user}^{max}-S_{gb}}{S_{gb}} = \epsilon, \nonumber \\
    S_{gb} = \frac{U_{user}^{max}}{1+\epsilon}.
    \label{half}
\end{align}

Since $ U_{user}^{max}-S_{gb} = \Delta U_{user}$, substituting this  into Eq.(\ref{half})  to get the following formulas.
\begin{align}
     \frac{\Delta U_{user}}{S_{gb}} &= \epsilon, \nonumber\\
     \Delta U_{user} &= \frac{\epsilon U_{user}^{max}}{1+\epsilon}, \nonumber\\
     \Delta U_{user} &= \frac{U_{user}^{max}}{1+\frac{1}{\epsilon}}.
\end{align}

When $\epsilon \rightarrow 0$, $\Delta U_{user} \rightarrow 0$. Since $\Delta U_{server} = \frac{1}{2} \Delta U_{user}$, it follows that $\Delta U_{server} \rightarrow 0$. Thus, we can conclude that when $U_{user}$ satisfies the stopping condition of the iteration and reaches its sub-maximum value, $U_{user}$ also reaches its sub-maximum value. Therefore, we only need to discuss one of the cases. If we let $U_{user}$ reach its sub-optimal value, then $U_{server}$ will also correspondingly reach its sub-optimal value.

\begin{table}[h]
    \centering
    \caption{$\Delta U$ for different $F_{server}$ \protect \footnotemark}
    \begin{tabular}{|c|c|c|c|c|c|}
         \hline  
        & \textbf{$\Delta U_{21}$}  &  \textbf{$\Delta U_{32}$}  & \textbf{$\Delta U_{43}$} & \textbf{$\Delta U_{54}$}  &  \textbf{$\Delta U_{65}$} \\ 
        \hline  
        $U_{user}$&5.4067 & 1.8032 & 0.9011 & 0.5407&0.3604 \\ 
        \hline 
       $U_{server}$& 2.70336 & 0.90112 & 0.45056 & 0.27034 &0.18022 \\
        \hline
       
    \end{tabular}
    \label{tabel_5}
\end{table}
\footnotetext[7]{$\Delta U_{21}$ represents the increase in utility from point (1, $U$) to point (2, $U$), and so on.}
\end{appendices}

\vspace{-33pt}
\begin{IEEEbiography}[{\includegraphics[width=1in,height=1.25in,clip,keepaspectratio]{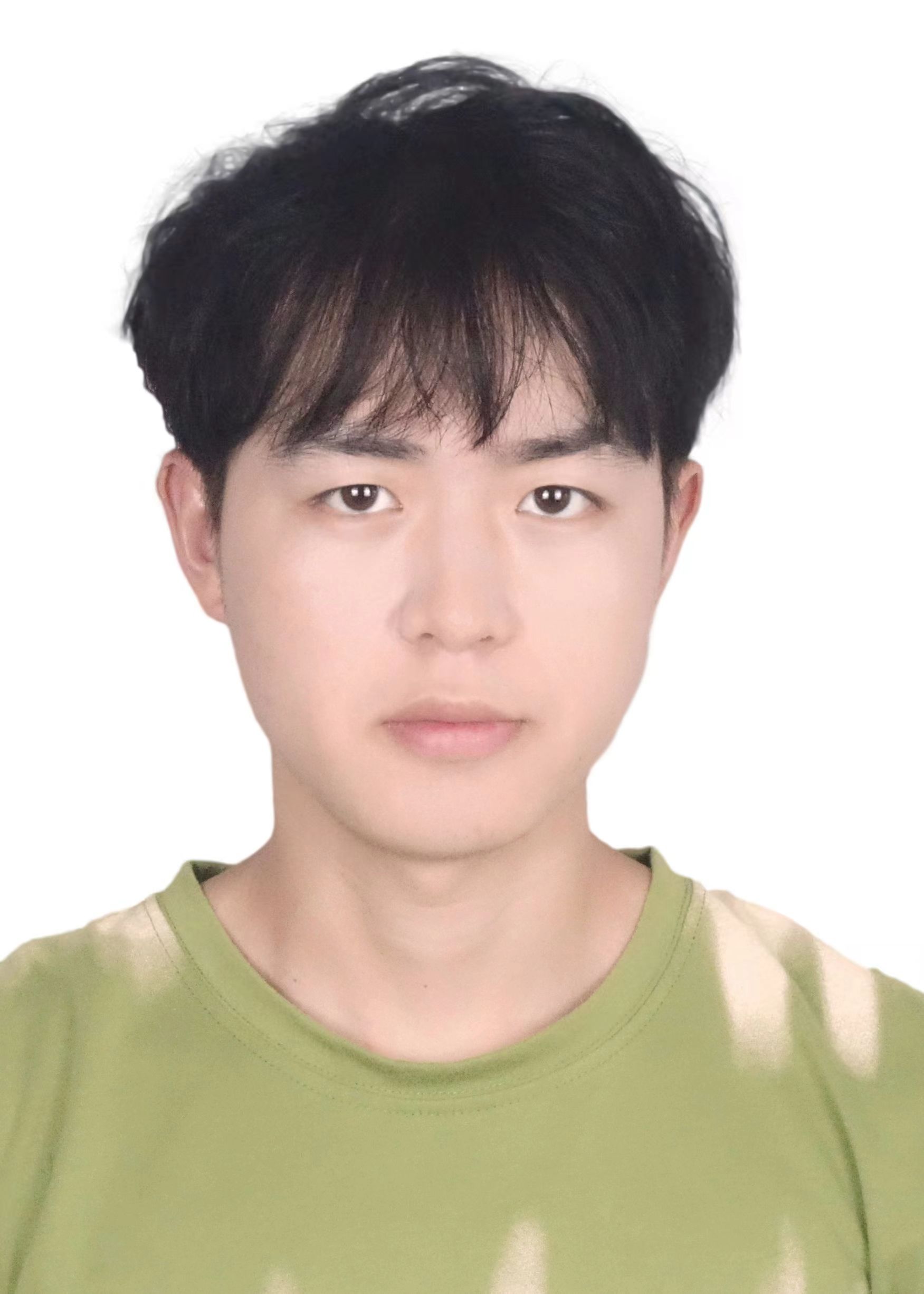}}]{Yun Xia} received the B.S. degree from Shanghai Ocean University, Shanghai, China, in 2022. He is currently working toward the master's degree in computer science and technology with the University of Shanghai for Science and Technology, Shanghai, China. His research focuses on mobile edge computing and pricing strategy.
\end{IEEEbiography}

\vspace{-33pt}
\begin{IEEEbiography}[{\includegraphics[width=1in,height=1.25in,clip,keepaspectratio]{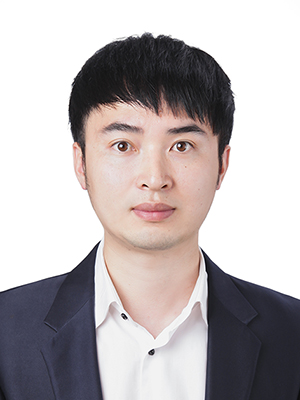}}]{Hai Xue} (Member, IEEE) received the B.S. degree from Konkuk University, Seoul, South Korea, in 2014, the M.S. degree from Hanyang University, Seoul, South Korea, in 2016, and the Ph.D. degree from Sungkyunkwan University, Suwon, South Korea, in 2020. He is currently an assistant professor at the School of Optical-Electrical and Computer Engineering, University of Shanghai for Science and Technology (USST), Shanghai, China. Prior to joining USST, he was a research professor with the Mobile Network and Communications Lab. at Korea University from 2020 to 2021, Seoul, South Korea. His research interests include edge computing/intelligence, SDN/NFV, machine learning, and swarm intelligence.
\end{IEEEbiography}

\vspace{-33pt}
\begin{IEEEbiography}[{\includegraphics[width=1in,height=1.25in,clip,keepaspectratio]{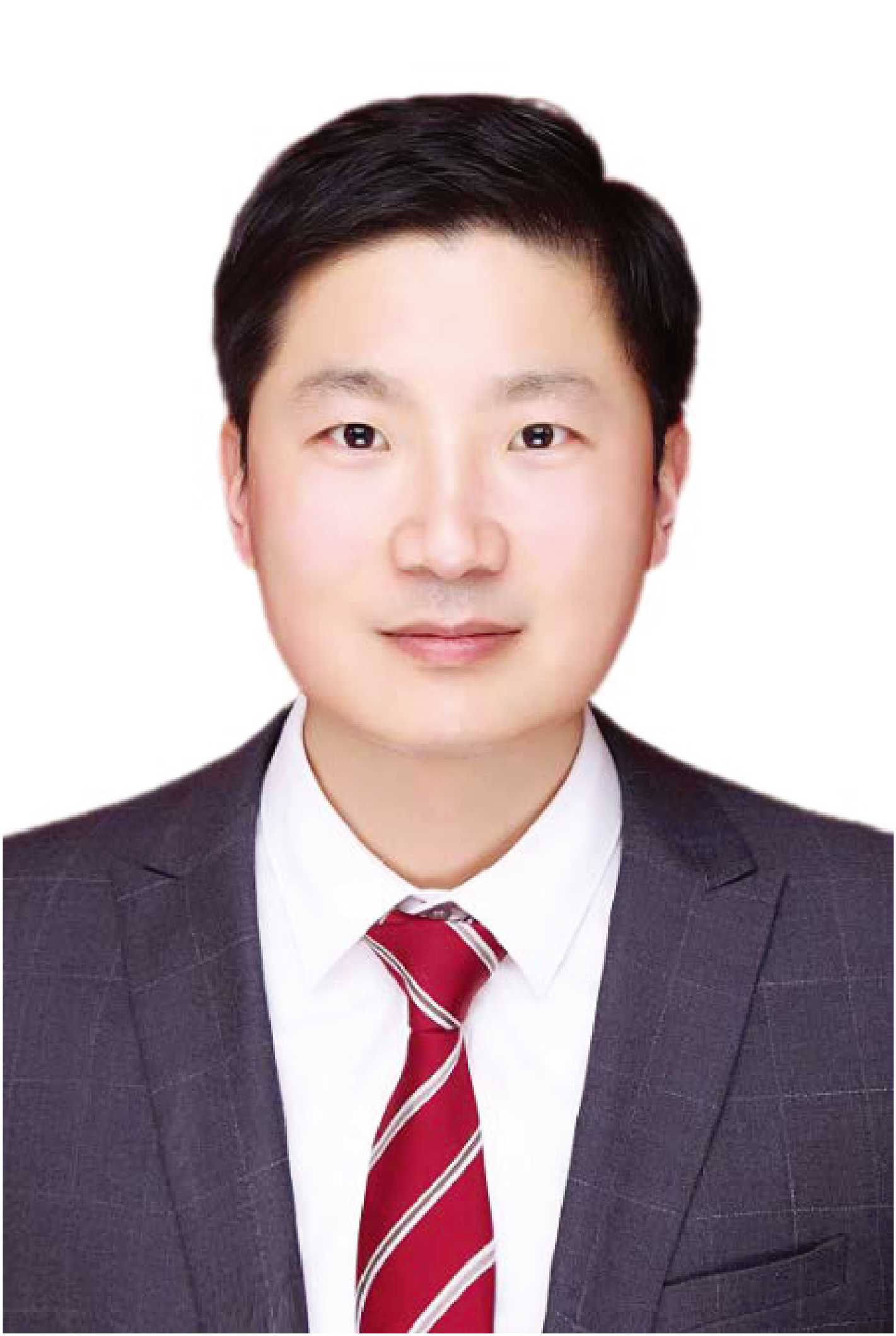}}]{Di Zhang} (Senior Member, IEEE) is currently an Associate Professor with Zhengzhou University. His research interests include wireless communications and networking, especially the short packet communications and its applications. He received the ITU Young Author Recognition in 2019, the First Prize Award for Science and Technology Progress of Henan Province in 2023, and the First Prize Award for Science and Technology Achievements from the Henan Department of Education in 2023. He is serving as an Area Editor for KSII Transactions on Internet and Information Systems, and has served as the Guest Editor for IEEE WIRELESS COMMUNICATIONS and IEEE NETWORK.
\end{IEEEbiography}

\vspace{-33pt}
\begin{IEEEbiography}[{\includegraphics[width=1in,height=1.25in,clip,keepaspectratio]{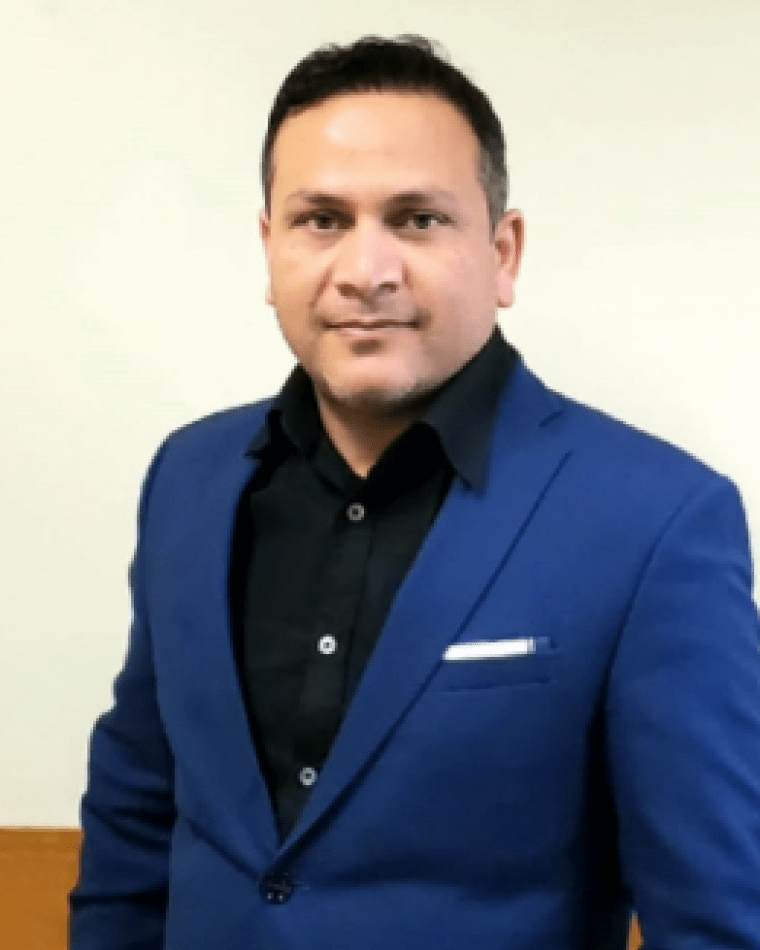}}]{Shahid Mumtaz} (Senior Member, IEEE) is an IET Fellow, the IEEE ComSoc Lecturer, and an ACM Distinguished Speaker. He has authored 4 technical books, 12 book chapters, and more than 300 technical articles (more than 200 IEEE Journals/Transactions, more than 100 conferences, and two IEEE best paper awards) in mobile communications. Dr. Mumtaz is the recipient of the NSFC Researcher Fund for Young Scientist in 2017 from China, and the IEEE ComSoC Young Researcher Award in 2020. He was awarded an Alain Bensoussan Fellowship in 2012. He is the Founder and EiC of IET Journal of quantum Communication, the Vice Chair of Europe/Africa Region—IEEE ComSoc: Green Communications \& Computing society, and the IEEE standard on P1932.1: Standard for Licensed/Unlicensed Spectrum Interoperability in Wireless Mobile Networks.
\end{IEEEbiography}

\vspace{-33pt}
\begin{IEEEbiography}[{\includegraphics[width=1in,height=1.25in,clip,keepaspectratio]{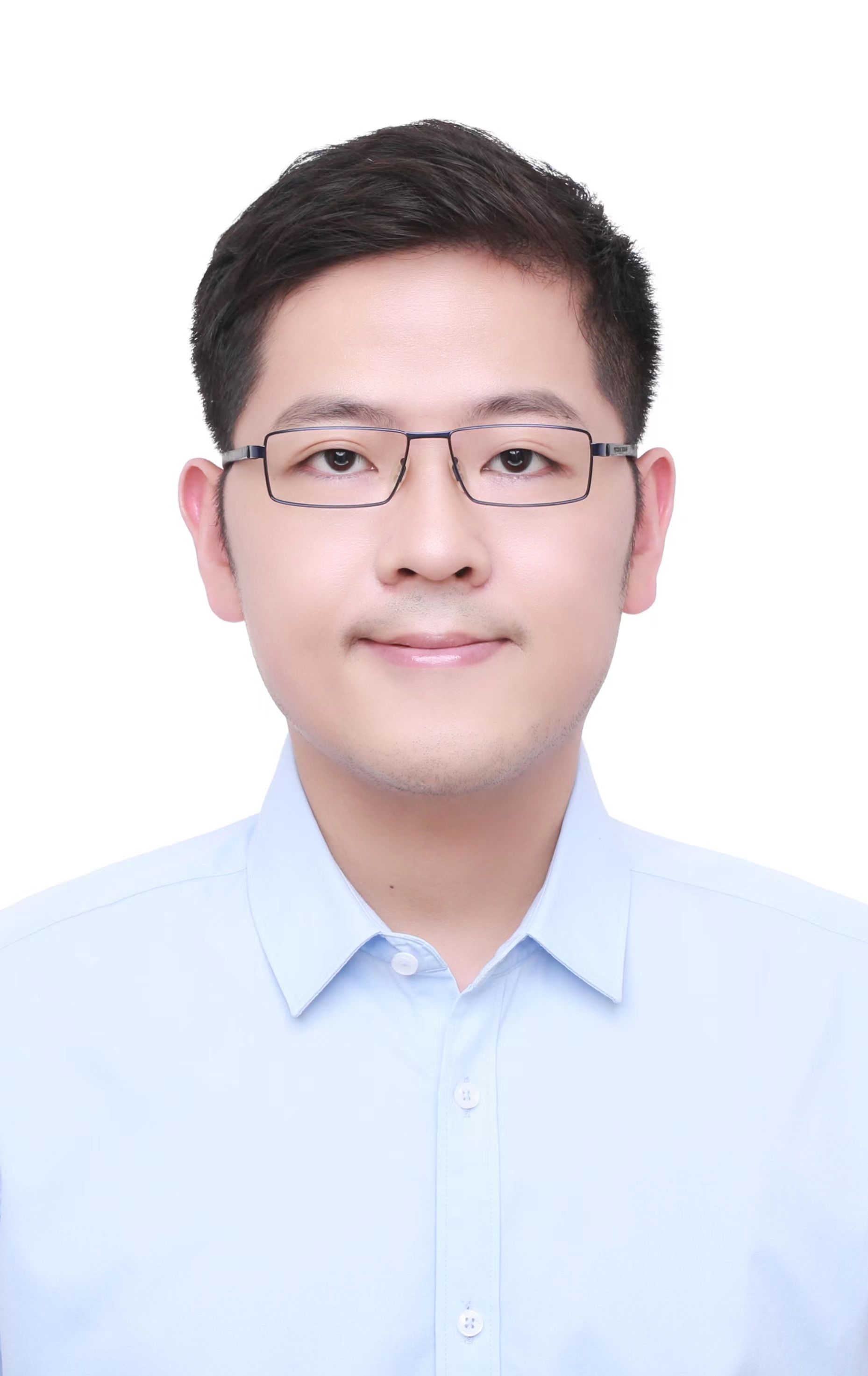}}]{Xiaolong Xu} (Senior Member, IEEE) received the Ph.D. degree in computer science and technology from Nanjing University, China, in 2016. He is currently a Full Professor with the School of Software, Nanjing University of Information Science and Technology. He has published more than 100 peer-review articles in international journals and conferences, including the IEEE TKDE, IEEE TPDS, JSAC, IEEE TSC, SCIS, IEEE TFS, IEEE T-ITS, ACM SIGIR, IJCAI, ICDM, ICWS, ICSOC, etc. He was selected as the Highly Cited Researcher of Clarivate (2021-2023). He received best paper awards from Tsinghua Science and Technology at 2023, Journal of Network and Computer Applications at 2022, and several conferences, including IEEE HPCC 2023, IEEE ISPA 2022, IEEE CyberSciTech 2021, IEEE CPSCom2020, etc. His research interests include edge computing, the Internet of Things (IoT), cloud computing, and big data.
\end{IEEEbiography}

\vspace{-33pt}
\begin{IEEEbiography}[{\includegraphics[width=1in,height=1.25in,clip,keepaspectratio]{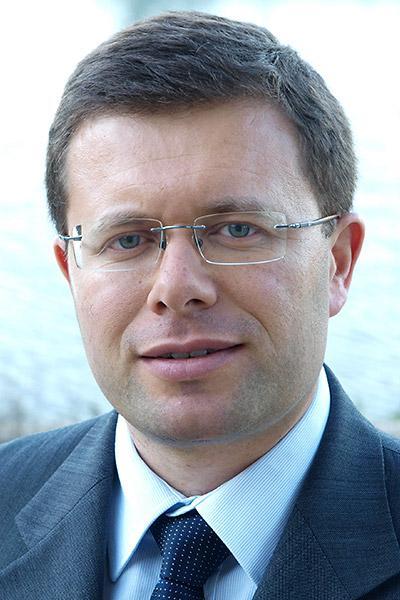}}]{Joel J. P. C. Rodrigues} (Fellow, IEEE) is currently the Leader of the Center for Intelligence, Fecomrcio/CE, Brazil, and a Full Professor with COPELABS, Lusófona University, Lisbon, Portugal. He is also a Highly Cited Researcher (Clarivate), No. 1 of the top scientists in computer science in Brazil (Research.com), the Leader of the Next Generation Networks and Applications (NetGNA) Research Group (CNPq), a Member Representative of the IEEE Communications Society on the IEEE Biometrics Council, and the President of the Scientific Council at ParkUrbis Science and Technology Park. He has authored or coauthored about 1150 papers in refereed international journals and conferences, three books, two patents, and one ITU-T Recommendation. He was the Director of Conference Development—IEEE ComSoc Board of Governors, an IEEE Distinguished Lecturer, a Technical Activities Committee Chair of the IEEE ComSoc Latin America Region Board, the past Chair of the IEEE ComSoc Technical Committee (TC) on eHealth and the TC on Communications Software, a Steering Committee Member of the IEEE Life Sciences Technical Community, and the publications co-chair. He is a member of the Internet Society, a senior member ACM, and a fellow of AAIA. He had been awarded several Outstanding Leadership and Outstanding Service Awards by IEEE Communications Society and several best papers awards. He is the Editor-in-Chief of the \textit{International Journal of E-Health and Medical Communications} and an editorial board member of several high reputed journals.
\end{IEEEbiography}

\end{document}